\documentclass[12pt]{article}
\usepackage{latexsym}
\def\mymedskip{\vskip\medskipamount}
\def\mymedbreak{\par \ifdim\lastskip<\medskipamount
  \removelastskip \penalty-100 \mymedskip \fi}
\def\myaftermedspace{\par \ifdim\lastskip<\medskipamount
  \removelastskip \penalty55\mymedskip\fi}
\newcommand{\eop}{{\unskip\nobreak\hfil\penalty50
          \hskip2em\hbox{}\nobreak\hfil$\Box$
          \parfillskip=0pt \finalhyphendemerits=0 \par}}
\newenvironment{proof}%
{\mymedbreak{\noindent\bf Proof:\enspace}}{\eop\myaftermedspace}
{\mymedbreak{\noindent\bf Proof of Theorem #1:\enspace}}{\eop\myaftermedspace}
\newenvironment{example}%
{\mymedbreak\refstepcounter{Exc}
                      {\em Example \theExc:}\enspace}%
{\eop\myaftermedspace}
\newtheorem{teor}{Theorem}[section]
\newcounter{Exc}

\newtheorem{lem}[teor]{Lemma}
\newtheorem{cor}[teor]{Corollary}

\newtheorem{rem}[teor]{Remark}
\newcommand{\beq}{\begin{equation}}
\newcommand{\eeq}{\end{equation}}
\newcommand{\beql}[1]{\begin{equation} \label{#1}}
\newcommand{\eeql}{\end{equation}}
\newcommand{\beqa}{\begin{eqnarray*}}
\newcommand{\eeqa}{\end{eqnarray*}}
\newcommand{\beqal}[1]{\begin{eqnarray} \label{#1}}
\newcommand{\eeqal}{\end{eqnarray}}
\newcommand{\beqan}{\begin{eqnarray}}
\newcommand{\eeqan}{\end{eqnarray}}
\newcommand{\bpf}{\begin{proof}}
\newcommand{\epf}{\end{proof}}

\newcommand{\bex}[1]{\begin{example} \label{#1}}
\newcommand{\eex}{\end{example}}

\newcommand{\G}{{\cal G}}

\newcommand{\cC}{{\cal C}}

\newcommand{\cO}{{\cal O}}

\newcommand{\bF}{{\bf F}}
\newcommand{\bK}{{\bf K}}
\newcommand{\bL}{{\bf L}}

\newcommand{\bZ}{{\bf Z}}

\newcommand{\per}{{\rm per}}
\newcommand{\chr}{{\rm char}}

\newcommand{\GF}{{\rm GF}}

\newcommand{\PG}{{\rm PG}}
\newcommand{\PGL}{{\rm PGL}}
\newcommand{\GL}{{\rm GL}}
\newcommand{\PSL}{{\rm PSL}}

\newcommand{\Tr}{{\rm Tr}}

\newcommand{\gs}{\sigma}

\newcommand{\gth}{\theta}
\newcommand{\gw}{\omega}
\newcommand{\tce}{\tilde{c}}
\newcommand{\tcC}{\tilde{\cC}}
\newcommand{\tf}{\tilde{f}}
\newcommand{\tgw}{\tilde{\omega}}
\newcommand{\txi}{\tilde{\xi}}
\newcommand{\tgs}{\tilde{\sigma}}
\newcommand{\tT}{\tilde{T}}

\newcommand{\ga}{\alpha}
\newcommand{\gb}{\beta}
\newcommand{\gc}{\gamma}
\newcommand{\gd}{\delta}

\newcommand{\gl}{\lambda}

\newcommand{\Gxi}{\mbox{$\langle\xi\rangle$}}
\newcommand{\Gphi}{\mbox{$\langle\phi\rangle$}}

\newcommand{\ord}{{\rm ord}}
\newcommand{\lcm}{{\rm lcm}}

\newcommand{\bFqn}{\GF(q_0)}
\newcommand{\bFqqn}{\GF(q_0^2)}

\begin{document}
\begin{titlepage}
\title{Nonstandard linear recurring sequence subgroups in finite fields and 
automorphisms of cyclic codes}
\date{\today}
\author{Henk D. L.\ Hollmann\\Philips Research Laboratories\\
Prof. Holstlaan 4, 5656 AA Eindhoven\\The Netherlands\\email: {\tt
henk.d.l.hollmann@philips.com}
}
\maketitle
\begin{abstract}
Let $q=p^r$ be a prime power, and let $f(x)=x^m-\gs_{m-1}x^{m-1}- \cdots
-\gs_1x-\gs_0$ be
an irreducible polynomial over the finite field $\GF(q)$ of size $q$.  A zero
$\xi$ of $f$
is called {\em nonstandard (of degree $m$) over $\GF(q)$\/} 
if the recurrence relation
$u_m=\gs_{m-1}u_{m-1} + \cdots + \gs_1u_1+\gs_0u_0$
with characteristic polynomial $f$
can generate the powers of $\xi$ in a nontrivial way, that is, with $u_0=1$ and
$f(u_1)\neq 0$. In 2003, Brison and Nogueira asked for a characterisation of
all
nonstandard cases in the case $m=2$, and solved this problem for $q$ a prime,
and later for $q=p^r$ with $r\leq4$.

In this paper, we first show that classifying nonstandard finite field elements
is equivalent to
classifying those cyclic codes over $\GF(q)$ generated by a single zero that
posses
extra permutation automorphisms. 

Apart from two sporadic examples of degree 11 over $\GF(2)$ and of degree 5
over $\GF(3)$, related to the Golay codes, 
there exist two classes of examples of nonstandard finite field elements.
One of these
classes (type I) involves irreducible polynomials $f$ of the form 
$f(x)=x^m-f_0$, and is well-understood.
The other class (type II) can be obtained from a primitive element
in some subfield by a process that we call extension and lifting.
We will use the known classification of the subgroups of $\PGL(2,q)$ 
in combination with a recent result by Brison and Nogueira to show 
that a nonstandard element of degree two over $\GF(q)$ necessarily is of type I
or type II, thus solving completely the classification problem for the case 
$m=2$.
\end{abstract}
\end{titlepage}

%
\section{\label{Sint}Introduction}
In a sequence of papers 
\cite{bn1,bn-comp,bn-mat,bn2,bn3}, 
Brison and Nogueira investigated when and how a
multiplicative subgroup $K$ of a finite field can be generated by a linear
recurrence relation of order $m$ with coefficients in a finite field $\GF(q)$,
with $q=p^r$ and $p$ prime. If the recurrence relation has characteristic
equation $f(x)\in\GF(q)[x]$, then such a subgroup is called an {\em 
$f$-subgroup\/}. In particular, they call such an $f$-subgroup {\em
non-standard\/} if it can be generated in a ``non-obvious'' way (that is, not 
as the sequence $1,\xi, \xi^2,\ldots$ with $\xi$ a zero of $f$).

Of particular interest is the case where the characteristic polynomial $f$ of
the recurrence relation is irreducible.
%
In this case, 
there are two known types of nonstandard $f$-subgroups. 
Here one type, referred to here as {\em type I\/}, 
arises from the ``degenerate'' case where $f$ is of the form
$f(x)=x^m-\eta$ with $\eta\in\GF(q)^*=\GF(q)\setminus\{0\}$. 
The other type, referred to here as {\em type II\/}, 
can be obtained from an
$f$-subgroup $K=\GF(q^m)^*=\GF(q^m)\setminus\{0\}$ (so with $f$ 
{\em primitive\/}) by a process that will be called
``lifting'' and ``extension'' in this paper.   
In the ``irreducible order two'' case where $f$ is irreducible of degree $m=2$,
Brison and Nogueira have
shown that there are no other examples if $q=p$ in\cite{bn2} and, recently, if 
$q=p^r$ with $r=2,3,4$ in \cite{bn3}.
In this paper we first show that cyclic codes of length $n$ over $\GF(q)$ 
generated by a single zero $\xi$ of degree $m$ over $\GF(q)$ that have ``extra''
permutation automorphisms provide examples of nonstandard $f$-subgroups, with
$f$ the minimal polynomial of $\xi$ over $\GF(q)$. (In fact, we will show that 
in the irreducible case, these two exceptional kind of objects are {\em
equivalent\/}.) As a consequence, we 
show that the binary and ternary Golay codes provide new nonstandard examples, 
of degrees $m=11$ and $m=5$, respectively.

The main results in this paper, when combined with a result from a preprint
\cite{bn3} by Brison and Nogueira, can be used to show that there are no other
examples in the ``irreducible order two'' case, for {\em any\/} $q$.
To explain our approach, we introduce a few definitions. 
%
An element $\xi$ in 
an extension field of $\GF(q)$ will be called {\em nonstandard of degree $m$
over $\GF(q)$\/} if its minimal polynomial $f$ over $\GF(q)$ has degree $m$ and
the subgroup $\Gxi$ generated by $\xi$ is a nonstandard $f$-subgroup. 
(It turns out that in this case {\em all\/} generators are nonstandard, with
the same degree and $q$-order as $\xi$.)
It can be shown that if $f$ is {\em irreducible\/} of degree $m$ over $\GF(q)$,
then {\em all\/} $f$-subgroups are of the form $\Gxi=\{1,\xi, \xi^2, \ldots\}$, 
for some zero $\xi$ of $f$ in $\GF(q^m)$. So in order to classify nonstandard
$f$-subgroups with $f$ irreducible, it is sufficient to classify nonstandard
finite field elements.

An important notion in this paper is the $q$-order $\ord_q(\xi)$
of an element $\xi$ in some
extension of $\GF(q)$, the smallest integer $d>0$ such that $\xi^d\in\GF(q)$.
There exist two processes, that we call ``extension'' and ``lifting'', which,
given a nonstandard $\phi$ of degree $m$ over a field $\GF(q_0)$, can be used 
to obtain a nonstandard $\xi$ of degree $m$ over an extension field $\GF(q)$, 
where $q=q_0^t$ and
$\gcd(m,t)=1$, with $\ord_q(\xi)=\ord_{q_0}(\phi)$ and
$\Gphi \subseteq \Gxi$.
The nonstandard examples of type II
are precisely the nonstandard elements of degree $m$ over $\GF(q)$ and
$q$-order $(q_0^m-1)/(q_0-1)$ that 
can be obtained from a primitive element of degree $m$ over $\GF(q_0)$ 
by lifting and extension.

Now, with each nonstandard finite field element of degree $m$ over $\GF(q)$ and
$q$-order $d$, we can associate a subgroup $\Xi$ in $PGL(m,q)$ which, in the
natural action on $\PG(m-1,q)$, has an orbit of size $d$.
In the case where $m=2$, the properties of this group $\Xi$ together with the
known classification of the subgroups of $\PGL(2,q)$ can be used to show that
$\Xi$ is actually equal to some subgroup $\PGL(2,q_0)$ or $\PSL(2,q_0)$ of
$\PGL(2,q)$, so that $d=q_0+1$, where $q=q_0^t$ with $t$ odd. 
Using this, we construct a nonstandard element 
$\phi$ of degree two over $\GF(q_0)$, of $q_0$-order $q_0+1$, from which $\xi$
can be obtained by lifting and extension.

Now a recent result from Brison and Nogueira \cite{bn3} states that if $\phi$
is nonstandard of degree two over $\GF(q_0)$ and has $q_0$-order $q_0+1$, then
$\phi$ must be primitive. As a consequence, in the above situation, we can 
conclude that the nonstandard $\xi$ is an known example, of the second type.

The contents of this paper are as follows. In Section~\ref{Spre}, 
we first introduce 
the problem in more detail. We discuss some well-known facts concerning linear
recurrence relations and linear recurring sequences, and use these to redefine
the notion of nonstandard finite field elements in terms of linearized
polynomials (or $q$-polynomials).  We describe the calls of examples of type I,
and we show that, with a few exceptions, a primitive element is also 
nonstandard.

In Section~\ref{Sauto}, we show that the classification problem for nonstandard
finite field elements is in fact equivalent to the problem of classifying the
cyclic codes with a single defining zero that have ``extra'' permutation
automorphisms.  

The methods  of lifting and extension to obtain new nonstandard elements from
old ones are introduced in Section~\ref{Slift}. 
We illustrate these techniques by constructing a class of examples referred to
as type II, from primitive elements in a subfield. 

In Section~\ref{Sgrm}, we first use the {\em companion matrix\/} of an
irreducible polynomial $f$ of degree $m$ over a field $\GF(q)$ to show that the
$q$-order of a zero $\xi$ of $f$ actually equals the {\em restricted period\/}
of~$f$. Then, if $\xi$ is also nonstandard, the companion matrix and
another matrix, considered as elements of $\PGL(m,q)$, generate a subgroup 
$\Xi$ of $\PGL(m,q)$ that has an orbit of size $d$ on $\PG(m-1,q)$.

In the remainder of the paper, 
we investigate this group $\Xi$ in the case where $m=2$. 
First, in~Section~\ref{Sgr2} we consider the case of small $q$-order 3, 4, or 5.
Then, in Section~\ref{Sgroup} we use these results together with the known 
classification of subgroups of $\PGL(2,q)$ to show that $\Xi$ is a subgroup
$\PGL(2,q_0)$ of $\PSL(2,q_0)$, where $q=q_0^t$ with $t$ odd, and the $q$-order
$d$ equals $q_0+1$. Finally, we establish the existence of a nonstandard
$\phi$ of degree two over $\GF(q_0)$, with $q_0$-order $q_0+1$, 
from which the original nonstandard $\xi$ can be obtained by lifting and 
extension. 
Now a recent result by Brison and Nogueira \cite{bn3} states that a nonstandard
element $\phi$ of degree two over $\GF(q_0)$ and with $q_0$-order $q_0+1$ is
necessarily primitive in $\GF(q_0^2)$, that is, has order $q_0^2-1$; as a
consequence, $\xi$ must be of type II.
\section{\label{Spre}Preliminaries}
Let $\bF$ be a field.
We will write $\bF^*=\bF\setminus\{0\}$ to denote the nonzero elements 
in~$\bF$. The collection of polynomials in $x$ with coefficients in $\bF$ will
be denoted by $\bF[x]$.
Consider the (homogeneous linear) recurrence relation
\beql{Erec} u_k = \gs_{m-1}u_{k-1}+\cdots + \gs_1 u_{k-m+1}+\gs_0 u_{k-m},
\eeql
where $\gs_0\in\bF^*$ and $\gs_1, \ldots,\gs_{m-1}\in\bF$. For later use, we
define $\gs_m=-1$.
Such a recurrence relation generates for any given sequence $u_0, \ldots,
u_{m-1}$ in an extension field $\bL\supseteq\bF$ of $\bF$
an ({\em $m$th order\/}) (homogeneous) {\em linear recurring sequence\/} 
$u=u(u_0, \ldots, u_{m-1})$
{\em in $\bL$\/}. 
The (monic) polynomial 
\beql{Ef} f(x) = x^m -\gs_{m-1}x^{m-1}-\cdots -\gs_1x-\gs_0\eeql
in $\bF[x]$ is called the {\em characteristic polynomial\/} of the recurrent
relation (\ref{Erec}); it has degree $\deg(f)=m$.
We will sometimes refer to a sequence $u=\{u_k\}_{k\geq0}$ satisfying
a recurrence relation (\ref{Erec}) with characteristic polynomial $f$ as an
{\em $f$-sequence\/}.

For later use, we state some crucial facts concerning linear recurring
sequences that we need later on. To this end, we need a few definitions. 
A {\em period\/} of a linear
recurring sequence $u$ is a positive integer $n$ for which $u_{k+n}=u_k$ holds
for all $k\geq0$; the smallest such number is called the {\em smallest
period\/} of the sequence, and will be denoted by $\per(u)$. 
The order $\ord(f)$ of a polynomial $f$ is the
smallest positive integer $N$ for which $f(x)$ divides $x^N-1$; if no such $N$
exists then we define $\ord(f)=\infty$. If $\xi$ is a nonzero element in some
extension $\bL$ of $\bF$, then we write
\beql{EGxi} \Gxi = \{1,\xi, \xi^2, \ldots\} \eeql
to denote the (multiplicative) group generated by $\xi$.
The order $\ord(\xi)$ is the smallest positive integer $n\geq0$ for which
$\xi^n=1$; if no such $n$ exists then $\ord(\xi)=\infty$. So we have that
$\ord(\xi)=|\Gxi|$.
%
\begin{teor}\label{Trec} Let $\bL\supseteq \bF$ be fields, and
let $u=u_0, u_1, \ldots, $ be a linear recurring sequence in $\bL$ satisfying
a recurrence relation (\ref{Erec}) with characteristic polynomial $f$ as in
(\ref{Ef}). Suppose that $f$ in $\bF[x]$, with $\gs_0\neq0$, and let
$\ord(f)<\infty$.
Then $\per(u)|ord(f)$.
Moreover, if $f$ has $m$ {\em distinct\/} zeros $\xi_0, \ldots, \xi_{m-1}$, 
then we have the following.\\
(i) The order $\ord(f)$ satisfy $\ord(f)=\lcm(\ord(\xi_i)\mid i=0,\ldots,
m-1)$; moreover, if each zero $\xi_i$ of $f$ has the same order $n$, then
$n=\ord(f)=\per(u)$ for each solution $u$ of (\ref{Erec}) with $(u_0, \ldots,
u_{m-1})\neq (0, \ldots, 0)$.\\
(ii) Suppose that $\bL$ contains all these distinct 
zeros of $f$. Then every solution $u$ 
of (\ref{Erec}) in $\bL$ can be written {\em uniquely\/} as 
\[u_k=L_0\xi_1^k + \cdots + L_{m-1}\xi_{m-1}^k\]
(for all $k\geq0$) with $L_0, \ldots L_{m-1}\in\bL$.
\end{teor}
\bpf For completeness' sake, we sketch a quick proof. First, if $u$ is
a solution of (\ref{Erec}), then since $\gs_0\neq0$ we may assume $u_k$ to be
defined {\em for all\/} integers $k$, and
\[ f(x)\sum_{k=-\infty}^\infty u_{-k}x^k = 0.\]
If $q(x)f(x)=x^N-1$, then multiplying both sides of the above relation by
$q(x)$ immediately shows that $N$ is a period of $u$;
hence $\per(u)|N$. 
Next, if $n=\per(u)$, then, writing $u^*(x)=u_0x^{n-1}+u_1x^{n-2}+\cdots +
u_{n-1}$ and $\gs_m=-1$, we have that
\[ f(x)u^*(x)=(1-x^n)h(x),\]
where 
\[h(x)=\sum_{j=0}^{m-1}\sum_{i=0}^{m-1-j}\gs_{i+j+1}u_ix^j\]
is a polynomial of degree at most $m-1$.
So if $\ord(\xi_i)=N=\per(f)>n$ holds for all $i$, then $\xi_i^n\neq0$, hence
$h(\xi_i)=0$, for all $i=0, \ldots, m-1$, 
which is not possible since $h$ has degree less than $m$.

Finally, if $\xi$ is a zero of $f$ in $\bL$, then $u_k=\xi^k$ defines a 
solution 
to (\ref{Erec}) in $\bL$. So obviously each $\bL$-linear combination of these 
$m$ solutions is also a solution in $\bL$.
Now the statement follows from the observation that $L_0, \ldots, L_{m-1}$ can
be uniquely determined from $(u_0, \ldots, u_{m-1})$ in terms of a Vandermonde
matrix over $\bL$.
\epf
\begin{rem}
In \cite{bn1}, it was claimed that if $f$ is irreducible over a finite field
$\bF$, then as a consequence of \cite{ln}, Theorem~8.28, each nonzero solution 
of the
recurrence relation (\ref{Erec}) with characteristic equation $f$ has smallest
period $\ord(f)$. However, the cited theorem only claims this to hold for 
solutions {\em in $\bF$\/}. 
The above proof, which, by the way, involves the same elements as the proof of 
Theroem~8.28 in \cite{ln}, shows that this also holds for solutions in an 
{\em extension\/} of~$\bF$. 
\end{rem}
%


In \cite{bn1}, a finite 
multiplicative subgroup $K$ of some extension $\bL$ of $\bF$
is called an {\em $f$-subgroup\/} if
it can be generated without repetitions by the recurrence relation (\ref{Erec})
with characteristic polynomial $f$.
That is, $K$ is an $f$-subgroup if there is a choice of $u_0, \ldots, u_{m-1}$ 
in $K$ such that the recurring sequence 
$u=u_0,u_1,\ldots$ generated by (\ref{Erec}) has (smallest) period $|K|$ and
$K=\{u_0, \ldots, u_{|K|-1}\}$. Note that we may assume without loss of 
generality that $u_0=1$ by dividing all members of the sequence $u$ by $u_0$, 
if necessary.
We say that $K$ is an {\em $m$th order linear recurring sequence subgroup\/}
if there is an $f$ of degree $m$ as in (\ref{Ef}) with $\gs_0\neq 0$ such that
$K$ is an $f$-subgroup.

For later use, we note the following.
For all fields $\bL$, a finite subgroup $K$ of $\bL^*$ is necessarily 
{\em cyclic\/}, see for example \cite{}. 
So if $|K|=n$, then $K$ consists precisely of the $n$ solutions of the equation
$x^n-1$, which must therefore all be distinct.
We conclude that for a given field $\bF$, there exists a 
{\em unique\/} subgroup $K$ of order $n$ in some extension of $\bF$
precisely when the characteristic $\chr(\bF)$ of $\bF$ satisfies 
$(n,\chr(\bF))=1$. In that case $K$ is cyclic, of the form
$K=\Gxi$, where $\xi$ is a primitive $n$th root of unity in an extension $\bL$
of $\bF$.

If $\xi$ a zero of a polynomial $f(x)\in\bF[x]$, 
then the sequence $u_k=\xi^k$ satisfies the recurrence relation (\ref{Erec}),
and hence $K=\Gxi$ is an $f$-subgroup. In \cite{bn1}, an $f$-subgroup $K$ with
$f$ of degree $m=\deg(f)=2$ was called {\em nonstandard\/} if $K$ can be 
generated by a solution $u$ of (\ref{Erec}) with smallest period $|K|$ for 
which $u_0=1$ and $u_1$ is {\em not\/} a
zero of $f$. Here, we extend this to the case of general degree, by calling an
$f$-subgroup $K=\Gxi$ {\em nonstandard\/} if $K$ can be generated by a solution
$u$ of (\ref{Erec}) with smallest period $|K|$ for which $u_0=1$ and 
\[ (u_0, u_1,\ldots, u_{m-1}) \neq (1,\xi, \ldots, \xi^{m-1})\]
for all zeros $\xi$ of $f$.
An $f$-subgroup that is not nonstandard is called {\em
standard\/}.
\begin{teor}\label{Tirr}
If $f$ is irreducible over $\bF$, if $\ord(f)<\infty$, 
and if $f$ has no multiple zeros, then each
$f$-subgroup $K$ in an extension $\bL$ of $\bF$ is of the form $K=\Gxi$, for
some zero $\xi$ of $f$ in $\bL$.
\end{teor}
\bpf
By Theorem~\ref{Trec}, under these assumptions all nonzero solutions $u$ of the
recurrence relation with characteristic equation $f$ have smallest period
$n=\ord(f)=\ord(\xi)$, for any zero $\xi$ of $f$. So an $f$-subgroup is cyclic 
of size~$n$, and since it is unique it must be equal to the group $\Gxi$.
\epf
\begin{rem} \label{Rfsub}
As stated in \cite{bn1}, even when $f$ is not irreducible,
no $f$-subgroup is known that is not of the form $\Gxi$ for a zero $\xi$ of
$f$, but it has not been proved that this must hold in general.
\end{rem}

In this paper, we will be interested in nonstandard $f$-subgroups. In view of
the preceeding remarks and observations, it seems reasonable to somewhat 
restrict our attention.

{\bf From now on, we will assume that $\bF$ is a finite field $\GF(q)$ with
$q=p^r$ and $p$ prime, and that $f$ is irreducible over $\bF$.}

If $f$ is irreducible of degree $m$ over $\GF(q)$, then $f$ has zeroes 
\beql{Ezeroes} \xi, \xi^q, \ldots, \xi^{q^{m-1}},\eeql
for some $\xi\in\bF_{q^m}$, of order $n=\ord(f)$ dividing $q^m-1$.
Of course all zeros of $f$ generate the same group $\Gxi$, which is an
$f$-subgroup.
So in view of Theorem~\ref{Tirr}, the following definition makes sense.
We will say that an element $\xi$ in some extension of $\GF(q)$ is {\em
nonstandard of degree $m$ over $\GF(q)$ and order $n=\ord(\xi)$\/} if the
minimal polynomial $f(x)$ of $\xi$ over $\GF(q)$ has degree $m$ and $\Gxi$ is a
nonstandard $f$-subgroup, of order (size) $n$. With this definition,
the clasification problem of nonstandard elements over $\GF(q)$ is equivalent
to the classification of nonstandard $f$-subgroups with $f$ irreducible over
$\GF(q)$. We will show later that if $f$ is irreducible of degree $m$ over
$\GF(q)$ and $K$ is a nonstandard $f$-subgroup of order $n$, then {\em all\/}
elements of order $n$ in $K$ (that is, all generators of $K$) are nonstandard
of degree $m$ over $\GF(q)$ (but with different minimal polynomials). Or,
stated differently, if $K$ is a nonstandard $f$-subgroup with $f$ irreducible
over $\GF(q)$, then $K$ is a nonstandard $g$-subgroup for {\em all\/} minimal
polynomials $g$ over $\GF(q)$ of generators of $K$.

The solutions of a recurrence relation for which the characteristic equation 
is irreducible can be described in terms of {\em linearized polynomials\/}, 
see, e.g., \cite{ln}, Chapter 8.
A {\em $q$-polynomial\/} 
of {$q$-order\/} $m$ over an extension field $\bL$
of $\GF(q)$ is a polynomial of the form
\[ L(x) = L_0x+L_1x^q +\cdots + L_{m-1}x^{q^{m-1}}\]
with coefficients $L_j$ in $\bL$ for $j=0, \ldots, m-1$, with $L_{m-1}\in\bL^*$.
Sometimes, a $q$-polynomial is also referred to as a {\em linearized\/} 
polynomial, if the value of $q$ is evident from the context. 
Note that such a polynomial is {\em $\bF_q$-linear\/}, that is,
\[ L(a x+b y)=aL(x)+bL(y)\]
for all $a,b\in\bF_q$.
We will call a $q$-polynomial {\em nonstandard\/} if it is not of the form
$L(x)=cx^{q^j}$ for some constant $c$ and some nonnegative integer $j$, and
{\em standard\/} otherwise.

\begin{teor}\label{Tqpol}
Let $\xi\in\GF(q^m)$ have minimal polynomial $f(x)$ over $\GF(q)$ as in
(\ref{Ef}).\\
(i) A sequence $u=\{u_k\}_{k\geq0}$ in $\GF(q^m)$ satisfies the linear recurring
relation (\ref{Erec}) with characteristic polynomial $f$ if and only if there
exists a $q$-polynomial $L(x)$ of $q$-order $m$ over $\GF(q^m)$ such that
$u_k=L(\xi^k)$ for all $k\geq0$.\\
(ii) We have that $\xi$ is nonstandard of degree $m$ over $\GF(q)$ if and only
if there exists a nonstandard $q$-polynomial $L(x)$ of $q$-order $m$ 
such that $L(\Gxi)=\Gxi$.
\end{teor}
\bpf
(i) 
Since $f$ is the minimal polynomial of $\xi$ over $\GF(q)$, we have that $f$ is
irreducible, with distinct zeros $\xi, \xi^q, \ldots, \xi^{q^{m-1}}$, all in
$\GF(q^m)$. So if a sequence $u=\{u_k\}_{k\geq0}$ in $\GF(q^m)$ satisfies
the recurrency (\ref{Erec}) with characteristic polynomial $f$, then 
according to Theorem~\ref{Trec}, there are $L_0, \ldots, L_{m-1}$ in $\GF(q^m)$
such that
\[ u_k = L_0\xi^k + L_1 \xi^{k q} + \cdots + L_{m-1}\xi^{k q^{m-1}}\]
for all $k\geq0$. 
So if we let $L(x)=L_0 x + L_1 x^q + \cdots + L_{m-1}x^{q^{m-1}}$, then $L(x)$
is a $q$-polynomial of $q$-order $m$ over $\GF(q^m)$ for which $u_k=L(\xi^k)$ 
for all $k\geq0$.

(ii) From part (i), we seen that the subgroup $\Gxi$ is generated by
a solution $u$ in $\GF(q^m)$ of the recurrence relation with characteristic 
polynomial~$f$ if and only if the $q$-polynomial $L(x)$ of $q$-order $m$
corrsponding to this 
solution $u$ satisfies $L(\Gxi)=\Gxi$. (Note that this can only happen if $L$ is
$q$-polynomial {\em over $\GF(q^m)$\/}.)
Now since $L$ is $q$-linear and since $1=\xi^0, \xi, \ldots, \xi^{m-1}$
constitute a basis for $\GF(q^m)$ over $\GF(q)$, the coefficients $L_0, \ldots,
L_{m-1}$ of $L$ are uniquely determined by the images $L(1), L(\xi),
\ldots,L(\xi^{m-1})$. By replacing $L(x)$ by $L'(x)=L(x)/L(1)$ if necessary
we may assume that $L(1)=1$. (Note that this does not change the 
``standardness'' of the $q$-polynomial at hand.) Then the standard
$q$-polynomials $L(x)=x^{q^j}$, $j=0, \ldots, m-1$, are precisely the 
$q$-polynomials that result in a ``standard'' generation of the $f$-subgroup
$\Gxi$ where $(u_0, \ldots, u_{m-1})=(1, \xi^{q^j}, \ldots, \xi^{(m-1)q^j})$.
\epf

Next, we will discuss two nonstandard examples.
Note that there are no nonstandard elements of degree $m=1$.

\bex{E2} The case where $m>1$ and
$\xi\in\bF_{q^m}^*$ has order $n>4$ and
minimal polynomial of the form $f(x)=x^m-\eta$
with $\eta=\xi^m\in\bF_q^*$ with $\eta\neq1$.
Note that if $q=p^r$ with $p$ prime, then 
$(p,m)=1$. Also, we must have $q>2$: indeed, 
if $q=2$, then $\eta=1$ is the only possibility, but since
$x^m-1=(x-1)(1+x+\cdots +x^{m-1})$ is reducible, this does not occur.
Under the above assumptions, $\xi$ has $q$-order $d=m$, and
\[ \langle\xi\rangle=\{1, \eta, \eta^2, \ldots, \eta^{e-1}\}\times 
\{1,\xi, \ldots, \xi^{m-1}\},\]
where $e=n/m$ is the order of $\eta$ and $n$ is the order of $\xi$.
Note that $e>1$, since if $e=1$, then $\eta=1$ and $x^m-1$ is not irreducible
for $m>1$.
Now let $\tau\in S_m$ be a permutation with $\tau(0)=0$,
and define
\[ L(\xi^j) = \eta_j \xi^{\tau(j)}\]
for $j=0, \ldots, m-1$.
Finally, extend $L$ by $\bF_q$-linearity to all of $\bF_{q^m}$.
Since $1, \xi, \ldots, \xi^{m-1}$ constitute a basis of $\bF_{q^m}$ over
$\bF_q$ and since $\tau$ is assumed to be a permutation, $L$ is well-defined
by $\bF_q$-linearity, and nonsingular on $\bF_{q^m}$. Hence, since
$L\langle\xi\rangle \subseteq \langle\xi\rangle$ by definition, we actually
have
equality here.

There are precisely $e^{m-1}(m-1)!$ possible $q$-polynomials $L$ with $L(1)=1$
and precisely $m$ forbiden (standard) ones. Hence if $e=2$, $m\geq3$ or
$e\geq3$, $m\geq2$, then some $L$ is nonstandard. This condition holds
precisely when $m\geq2$, $e>1$, and $n>4$.

In particular, it is easily verified that there is an example of degree 2 over
$\bF_q$ with order $n$ and $q$-order 2 if and only if
$n=2e>4$ and both $q$ and $(q-1)/e$ are odd.
We will refer to such examples as {\em examples of type I\/}.
\eex

\bex{E1} If $m>2$ or $m=2, q>2$, then a primitive element of $\bF_{q^m}$ is
nonstandard over $\bF_q$.
This is the case where $\xi\in\bF_{q^m}$ has order $n=q^m-1$, so that
$\langle\xi\rangle=\bF_{q^m}^*$,
where $\bF_{q^m}^*=\bF_{q^m}\setminus\{0\}$. Indeed, in that
case {\em any\/} $q$-polynomial $L\in\bF_{q^m}[x]$ that is nonsingular on
$\bF_{q^m}$ will fix $\bF_{q^m}^*$ as a set, so is nonstandard polynomial for
$\xi$ except when of the form $\xi^cx^{q^j}$ for some $c\in\{0,\ldots, q^m-2\}$
and some $j\in\{0,1,\ldots, m-1\}$.
Here, $L$ is called nonsingular if the associated $\bF_q$-linear map $L$ on
$\bF_{q^m}$ is nonsingular; equivalently, if $L(x)\neq0$ for
$x\in\bF_{q^m}^*$. Note that the requirement that $L$ is nonsingular is
necesary and sufficient for $L$ to act as a permutation on $\bF_{q^m}^*$.

Now there are precisely $(q^m-1)(q^m-q)\cdots (q^m-q^{m-1})$ nonsingular 
$\bF_q$-linear maps on $\bF_{q^m}$, which are all of the form of a
$q$-polynomial in $\bF_{q^m}[x]$. Precisely $m(q^m-1)$ of these are 
``forbidden'', but all others provide nonstandard $q$-polynomials. It is
easily shown that for integers $m,q\geq2$, we have
\[(q^m-1)(q^m-q)\cdots (q^m-q^{m-1})>m(q^m-1)\] 
except when $m=2$ and $q=2$.

We will see later that primitive elements are particular cases of a class of
examples referred to as {\em type II} examples.
\eex
\section{\label{Sauto}Automorphisms of cyclic codes}
In this section, we will show that the classification problem of nonstandard
elements over $\GF(q)$ is equivalent to the problem of determining which cyclic 
codes over $\GF(q)$ defined by a single zero have ``extra'' automorphisms.
We begin by a brief introduction to cyclic codes. For more details, see e.g.\
\cite{lint}.

We will denote by $S_n$ the collection of all permutations on 
$\{0,1,\ldots, n-1\}$. 
In what follows, we will slightly abuse notation and use the same symbol $\pi$
to denote both a permutation from $S_n$ 
and the induced permutation on the $n$-dimensional vectorspace $\GF(q)^n$ given
by 
\[\pi: c \mapsto c^\pi= (c_{\pi(0)}, \ldots, c_{\pi(n-1)}).\]

A {\em cyclic code of length $n$ over $\GF(q)$\/} 
is a $\GF(q)$-linear subspace of
$\GF(q)^n$ closed under the map
\[ \sigma: (c_0, \ldots, c_{n-1}) \mapsto (c_{n-1}, c_0, \ldots, c_{n-2}).\]
This map, as well as the underlying permutation 
\[ \sigma: i \mapsto i-1 \bmod n,\]
are both referred to as a {\em cyclic shift\/}.


In what follows, we will identify a vector $c=(c_0, \ldots,
c_{n-1})\in\GF(q)^n$ with its associated polynomial
\[ c(x)=\sum_{i=0}^{n-1} c_i x^i\]
in $\GF(q)[x]\bmod x^n-1$.
Note that the cyclic shift $c^\gs$ of a vector $c$ has corresponding polynomial
$c^\gs(x)=xc(x)$; so multiplication by $x$ in $\GF(q)[x]\bmod x^n-1$ correspond
to a cyclic shift.

Let $n|q^m-1$, and
let $Z\subseteq \bF_{q^m}^*$ be a collection of field 
elements of order dividing $n$, so that $\ga^n=1$ holds for all $\ga\in Z$. The
{\em cyclic code of length $n$ over $\GF(q)$ with defining zeroes $Z$\/} is the
collection $C=C(n,q,Z)$ of all $c=(c_0, \ldots, c_{n-1})\in\GF(q)^n$ for which 
\[ c(\ga)=\sum_{i=0}^{n-1} c_i \ga^i=0\]
holds for all $\ga\in Z$. We refer to an element $c\in C$ as a {\em code
word\/}. 
Note that if $c(x)$ is in $C$, then the cyclic shift $xc(x)$ 
is again in $C$; since a cyclic code is also linear it is in fact an 
{\em ideal\/} in $\GF(q)[x]\bmod x^n-1$. 

If $c(x)$ has all its coefficients in $\GF(q)$, then $c(x)^q = c(x^q)$. As a
consequence, the codes $C(n,q,Z)$ and $C(n,q,\bar{Z})$ are equal, where 
$\bar{Z}=\{z^{q^j} \mid z\in Z; i=0, \ldots m-1\}$. 

A permutation $\pi\in S_n$ is called a {\em permutation automorphism\/} of a 
cyclic code $C\subseteq \GF(q)^n$ if for all code words 
$c=(c_0, \ldots, c_{n-1})\in C$, the permuted word
\[ c^\pi=(c_{\pi(0)}, \ldots, c_{\pi(n-1)})\]
is again in $C$. 
Now 
\[c^\pi(\xi)=\sum_{i=0}^{n-1}c_{\pi(i)}\xi^i=\sum_{j=0}^{n-1}
c_j\xi^{\pi^{-1}(j)},\]
so beside the cyclic shift $\gs$ also the Frobenius permutation $\phi: i
\mapsto qi \bmod n$ is a permutation automorphism of a cyclic code of length
$n$ over $\GF(q)$.

The next theorem provides a relation between automorphisms $\pi$ of cyclic
codes and $q$-polynomials fixing sets $\Gxi$.
\begin{teor}\label{Tauto} 
Let $\xi\in \GF({q^m})^*$
have order $\ord(\xi)=n$ and degree $m$ over $\GF(q)$, and let
$C\subseteq\GF(q)^n$ be the cyclic code $C=C(n,q,\{\xi\})$ of length $n$ 
over $\GF(q)$ with defining zero $\xi$. 
Then a permutation $\pi\in S_n$ is a permutation automorphism of $C$ if and
only if the map $L: \xi^i\mapsto \xi^{\pi(i)}$ extends to a $q$-polynomial of
$q$-order $m$ over $\GF({q^m})$.
\end{teor}
\bpf
First, suppose that $L$ is a $q$-polynomial of $q$-degree $m$ 
that fixes \Gxi, and
let $\pi\in S_n$ be the permutation induced by $L$, that is, let $\pi$ be 
such that
$L(\xi^i)=\xi^{\pi(i)}$ 
for all $i=0, \ldots, n-1$.
Then if $c\in C$, we have
\[0 = L(0) = L(\sum_{i=0}^{n-1}c_i\xi^i) 
	= \sum_{i=0}^{n-1} c_i L(\xi^i)
	= \sum_{i=0}^{n-1} c_i \xi^{\pi(i)}
	= \sum_{j=0}^{n-1} c_{\pi^{-1}(j)} \xi^j
\]
that is, $c^{\pi^{-1}}$ is in $C$.
So $\pi^{-1}$, and hence also $\pi$, is a permutation  automorphism of $C$.

Conversely, let $\pi$ be a permutation automorphism of $C$.
We define a $q$-polynomial $L(x)=L_0x+\cdots+L_{m-1}x^{q^{m-1}}$ of $q$-order
$m$ by letting 
\beql{Ldef} L(\xi^j)=\xi^{\pi(j)} \eeql
for $j=0, \ldots, m-1$, and then extending $L$ to
all of $\GF(q^m)$ by $\GF(q)$-linearity. Note that since we assumed that
$\xi$ has degree $m$ over $\GF(q)$, we have that $1, \xi, \ldots, \xi^{m-1}$
constitute a basis for $\GF(q^m)$ over $\GF(q)$, so $L$ is uniquely determined.
We claim that now (\ref{Ldef}) holds for {\em all\/} $j=0, \ldots, n-1$. 
Indeed, let $j\geq m$. By our assumptions on $\xi$, 
there are $a_0, \ldots, a_{m-1}\in \bF_q$ such that
$\xi^j=a_0 +a_1 \xi + a_2 \xi^2 + \cdots + a_{m-1}\xi^{m-1}$.
Note that then
\beql{ELexp}
L(\xi^j)=a_0\xi^{\pi(0)} +a_1 \xi^{\pi(1)} + a_2 \xi^{\pi(2)} + \cdots + 
a_{m-1}\xi^{\pi(m-1)}.
\eeql
Now since $C$ is the code with defining zero $\xi$, the word
\[ c = (a_0, \ldots, a_{m-1}, 0, \ldots, 0, -1,0, \ldots, 0),\]
with 
\[
	c_i=\left\{ \begin{array}{ll}
                    a_i, & \mbox{if $0\leq i\leq m-1$;} \\
                    -1, & \mbox{if $i=j$;} \\
                    0, & \mbox{otherwise,}
                     \end{array}
            \right.
\]
is in $C$. So by our assumption that $\pi$, and hence also $\pi^{-1}$, 
is a permutation automorphism of $C$, the word $c^{\pi^{-1}}$ is also in $C$.
Hence we have that
\[ 0 = \sum_{i=0}^{n-1}c_i\xi^{\pi(i)}
	= -\xi^{\pi(j)} + a_0\xi^{\pi(0)}+ \cdots + a_{m-1}\xi^{\pi(m-1)}
	=  -\xi^{\pi(j)} +L(\xi^j).
\]
We conclude that $L(\xi^j)=\xi^{\pi(j)}$ holds for {\em all\/} $j=0, \ldots,
n-1$, 
as claimed.
\epf

In view of Theorem~\ref{Tqpol}, we immediately have the following consequence.
\begin{cor}\label{Cqpol}
Let $\xi\in \GF({q^m})^*$
have order $\ord(\xi)=n$ and degree $m$ over $\GF(q)$.
Then the cyclic code $C(n,q,\xi)$ 
of length $n$ over $\GF(q)$ with defining zero $\xi$ has
``extra permutation automorphisms'', that is, a permutation automorphism group
{\em stricktly larger\/} than the group $\langle \gs, \phi\rangle$ of order
$mn$ generated by
the cyclic shift and the Frobenius permutation, if and only if $\xi$ is
nonstandard of order $n$ and degree $m$ over $\GF(q)$.
\end{cor}

From the above corollary we can obtain two new examples of
nonstandard elements.

\bex{EG2} (Binary Golay) Let $q=2$, $n=23$, and $m=11$. Then $n|2^{11}-1$. 
Let $\ga$ be primitive in $\GF(2^{11})$, and let $\xi=\ga^{(2^{11}-1)/23}$.
Then $\xi$ is a primitive $23$-th root of unity in $\GF(2^{11})$.
The {\em binary Golay code\/} is the binary length $n=23$ code with defining 
zero $\xi$.
It can be shown that this code has minimum distance $7$ (in fact, it a {\em 
perfect\/} binary 3-error-correcting code). Its automorphism group is the
{\em Mathieu\/} group $M_{23}$, a simple group of order 200960, entirely
consisting of permutations. As a consequence of Corollary~\ref{Cqpol}, we
conclude that $\xi$ is nonstandard of order $n=23$
and degree $m=11$ over $GF(2)$.
Its $2$-order is $d=23>m$, and we see immediately that this provides an example
not of the form of the two known types.
\eex

\bex{EG3} (Ternary Golay) Let $q=3$, $n=11$, and $m=5$. Then $n|2^{m}-1$. 
Let $\ga$ be primitive in $\GF(2^{5})$, and let $\xi=\ga^{(2^{5}-1)/11}$.
Then $\xi$ is a primitive $11$-th root of unity in $\GF(2^{5})$.
The {\em ternary Golay code\/} is the ternary length $n=11$ code with defining 
zero $\xi$.
It can be shown that this code has minimum distance $5$ (in fact, it a {\em 
perfect\/} ternary 2-error-correcting code). Its automorphism group is twice the
{\em Mathieu\/} group $M_{11}$, a simple group of order 7920,
which itself consists entirely of permutations. 
As a consequence of Corollary~\ref{Cqpol}, we
conclude that $\xi$ is nonstandard of order $n=11$
and degree $m=5$ over $GF(3)$.
Its $3$-order is $d=11>m$, and we see immediately that this provides another 
example not of the form of the two known types.
\eex

Examples of cyclic codes with ``extra'' automorphisms seem to
be quite rare. For example, the only (binary) quadratic-residue codes
with ``extra'' automorphisms of length less than 4000
are the $(7,4,3)$ Hamming code and the binary
Golay code \cite{mcws}.  
\section{\label{Slift}Extension and lifting}
For later use, we investigate when we can conclude that a $q$-polynomial
$L$ in $\GF(q^m)[x]$ for some $\phi$ of degree $m$ that acts as a bijection on
some other subgroup $\Gxi$ of $\GF(q^m)^*$ is actually nonstandard for $\xi$.
The result is as follows.
\begin{lem}\label{LLstand}
Let $L$ be a $q$-polynomial in $\GF(q^m)[x]$, and let $\xi$ have degree $m$
over $\GF(q)$. If $L(\xi^i)=\xi^{iq^j}$ for $i=0, \ldots, m-1$, then
$L(x)=x^{q^j}$ on $\GF(q^m)$.
\end{lem}
\bpf
If $\xi$ has degree $m$ over $\GF(q)$, then $1,\xi, \ldots, \xi^{m-1}$
constitute a basis for $\GF(q^m)/\GF(q)$, hence a $q$-linear map $L$
on $\GF(q^m)$ is determined on $\GF(q^m)$ by the images on $1, \xi, \ldots,
\xi^{m-1}$. Also, if $L$ in in $\GF(q^m)[x]$, then $L$ is determined as a
polynomial by its action on $\GF(q^m)$.
\epf

We also need the following simple observation concerning degrees.
\begin{lem}\label{Ldeg}
(i) An element $\xi\in\GF(q^m)$ of order $n$ has degree $m$ over $\GF(q)$ if 
and only if $m$ is the smallest integer $t\geq1$ for which $n|q^t-1$. 

\noindent
(ii)
If $\phi$ has degree $m$ over $\GF(q)$ and $\xi\in\GF(q^m)$ has $\phi\in\Gxi$,
then $\xi$ also has degree $m$ over $\GF(q)$.
\end{lem}
\bpf
(i) If $n|q^t-1$, then $\Gxi=\gth^{(q^t-1)/n}\subseteq \GF(q^t)^*$.

\noindent 
(ii) If $\phi\in\Gxi$, then $\langle \phi \rangle \subseteq \Gxi$, hence the
order $n$ of $\phi$ divides the order of $\xi$. Now the result follows from
part (i).
\epf

The {\em order\/} $\ord(\xi)$ of an element $\xi\in\GF(q^m)$ of degree $m$
over $\GF(q)$ was defined as the smallest positive integer $n$ for which
$\xi^n=1$. We now define the {\em $q$-order\/} $\ord_q(\xi)$ as the smallest
positive integer $d$ for which $\xi^d\in\GF(q)$. The $q$-order is an important 
notion in this paper. 
It is related to another notion, the resticted period,
which was investigated in \cite{bn-mat} and
played an important role in \cite{bn2} and \cite{bn3}.
Here, the {\em restricted period\/} $\delta(f)$ of a polynomial
$f(x)\in\GF(q)[x]$ as in
(\ref{Ef}), with corresponding recurrence relation (\ref{Erec}), is the first
positive integer $n$ for which the solution $u=\{u_k\}_{k\geq0}$ of
(\ref{Erec}) with
\[ (u_0,u_1,\ldots, u_{m-2}, u_{m-1})= (0,0,\ldots, 0, 1)\]
satisfies
\[ (u_0,u_1,\ldots, u_{m-2}, u_{m-1}) = (0,0,\ldots, 0, \gl),\]
for some $\gl\in\GF(q)^*$. The next theorem states this relation.

\begin{teor}\label{Tresper} The $q$-order of an element $\xi$ in an extension
$\GF(q^m)$ of $\GF(q)$ is equal to the restricted period of its minimal
polynomial over $\GF(q)$.
\end{teor}
We will prove this theorem in Section~\ref{Sgrm}.
In the next theorem, we collect some important properties of the $q$-order.
\begin{teor}\label{Tqord}
Let $\xi\in\GF(q)$ have degree $m$ over $\GF(q)$, with order $n=\ord(\xi)$ and
$q$-order $d=\ord_q(\xi)$. \\
(i) We have $m\leq d$ and  $d|(q^m-1)/(q-1)$.\\
(ii) We have that
\[ d = n/(n,q-1)\]
and $n=de$, where $e=(n,q-1)$ satisfies $(d,(q-1)/e)=1$.
\end{teor}
\bpf
(i)
If $d=\ord_q(\xi)$ and $\xi^d=\eta \in\GF(q)$, then $\xi$ is a zero of the
polynomial $x^d-\eta$ in $\GF(q)[x]$.
Hence the minimal polynomial of $\xi$, of degree $m$ by our
assumptions, divides $x^d-\eta$, whence $m\leq d$.

The collection of all integers $k\geq0$ for which $\xi^k\in\GF(q)$ is an ideal,
hence is of the form $d\bZ$. Now since $\xi\in\GF(q^m)^*$, we have that
$\varphi=\xi^{(q^m-1)/(q-1)}$ satisfies $\varphi^{q-1}=1$, hence
$\varphi\in\GF(q)$. We conclude that $d|(q^m-1)/(q-1)$.

(ii)We have $\xi^d\in\GF(q)$ if and only if $\xi^{d(q-1)}=1$, that is, if and
only if $n|d(q-1)$, or, equivalently, if and only if $n/(n,q-1)$ divides
$d(q-1)/(n,q-1)$.
Since $n/(n,q-1)$ and $(q-1)/(n,q-1)$ are relatively prime, the latter happens
if and only if $n/(n,q-1)$ divides $d$.
If we now write $e=(n,q-1)$, then $n=de$ and now the condition on $e$ follows
from the expression for $d$.
\epf
We will refer to our next theorem as the {\em extension theorem\/}. It enables
us to extend a nonstandard subgroup to a bigger one.
\begin{teor}\label{Text} Let $\phi$ be nonstandard of degree $m$ over $\GF(q)$.
Then every $\xi\in\GF(q)^* \Gphi$ for which $\Gphi\subseteq \Gxi$
(so with $\xi=\gl \phi$ and $\phi=\xi^i$ for some $\gl\in\GF(q)^*$ and integer
$i$) is again nonstandard of degree $m$ over $\GF(q)$, with the same $q$-order 
as $\phi$;
moreover, every $q$-polynomial $L(x)$ of $q$-degree $m$ over $\GF(q^m)$ 
for which $L(\Gphi)=\Gphi$ satisfies $L(\Gxi)=\Gxi)$.
\end{teor}
\bpf
We begin by observing that $\GF(q)^* \Gphi$ is a multiplicative subgroup 
of~$\\GF(q^m)^*$; since $\GF(q^m)^*$ is cyclic, all its subgroups are also 
cyclic, and hence there exists an element $\gth\in\GF(q^m)^*$
such that $\GF(q)^* \Gphi=\langle \gth\rangle$. 

Now let $n=\ord(\phi)$ and $d=\ord_q(\phi)$ denote the order and $q$-order of
$\phi$, respectively.
According to Theorem~\ref{Tqord}, we have that $d=n/(n,q-1)$. Write 
$k=(q-1)/(n,q-1)$; for later use, we note that $(d,k)=1$ and $k|q-1$.
Now $\GF(q)^* \Gphi=\GF(q)^*\{1, \phi, \ldots, \phi^{d-1}\}$ has size $(q-1)d$,
hence $\gth$ has order $d(q-1)=nk$. So by Theorem~\ref{Tqord}, $\gth$ has 
$q$-order $\ord_q(\gth)= d(q-1)/(d(q-1),(q-1)) = d=\ord_q(\phi)$, so $\gth$ and
$\phi$ have the same $q$-order. 

Now let $L$ be a $q$-polynomial of $q$-degree $m$ over $\GF(q^m)$ that fixes
$\Gphi$, say with
$L(\phi^i)=\phi^{\pi(i)}$ for some permutation $\pi\in S_n$. We claim that
$L(\langle\gth \rangle)=\langle\gth \rangle$. To see this, first note that
if $\ga\in\GF(q)$, then $L(\ga\phi^i)=\ga L(\phi^i)$; hence 
$L(\langle\gth \rangle)=L(\GF(q)^* \Gphi)\subseteq \GF(q)^* \Gphi=\langle\gth
\rangle$. 
Now, suppose that $L(\ga \phi^i)=L(\gb\phi^j)$ for some $\ga,\gb\in\GF(q)^*$ and
some integers $i,j$. Then 
\[ \ga \phi^{\pi(i)} = \ga L(\phi^i)= L(\ga \phi^i)
= L(\gb\phi^j) = \gb L(\phi^j) = \gb\phi^{\pi(j)}, \]
and hence 
\[ \gc=\ga/\gb = \phi^{\pi(j)-\pi(i)}\in\GF(q)^*.\]
Therefore,
\[L(\phi^j)=\phi^{\pi(j)} = \gc \phi^{\pi(i)} 
= \gc L(\phi^i) = L(\gc \phi^i)=L(\phi^{\pi(j)-\pi(i)+i}).\]
Since $L(\Gphi)=\Gphi$, we conclude that 
\[ \phi^j = \phi^{\pi(j)-\pi(i)+i},\]
hence 
\[ \gc \phi^i = \phi^{\pi(j)-\pi(i)+i}=\xi^j,\]
so that $\ga\phi^i=\gb\phi^j$. We have shown that $L$ is one-to-one on
$\GF(q)^* \Gphi=\langle\gth \rangle$, and hence $L(\langle\gth
\rangle)=\langle\gth\rangle$. 

Next, suppose that $\Gphi\subseteq \Gxi\subseteq \GF(q)^* \Gxi=
\langle\gth\rangle$. 
Then $n=\ord(\phi)|\ord(\xi)$ and $\ord(\xi)|\ord(\gth)=nk$, hence there are 
integers $s,t$ with $k=st$ such that $\ord(\xi)=nk/s=nt$. Moreover,
since $(nt,q-1)=(n,q-1)(dt,k)=(n,q-1)t$, we conclude from Theorem~\ref{Tqord} 
that $\xi$ has $q$-order
$\ord_q(\xi)=nt/(nt,q-1)=d$. So $\xi$ and $\phi$ have the same $q$-order.

Finally, since $\Gphi\subseteq \Gxi$, we have that $\Gxi=K \Gphi$ with $K$ the
subgroup of $\GF(q)^*$, of order $nt/d=et$, where $e=(n,q-1)$
(since $nt=det | d(q-1)$, we have $et|(q-1)$ and such a subgroup $K$ does 
indeed exist);
in fact, we have $K=\Gxi\cap \GF(q)^*$.
To see this, first note that since $\phi^d\in\GF(q)^*$ has order 
$e=n/d=(n,q-1)$ and $K$ has order $et$, we have that $\langle \phi^d
\rangle\subseteq K$. Hence
\[K\Gphi=K\langle \phi^d \rangle \{1,\phi,\ldots,
\phi^{d-1}\}=K\{1,\phi,\ldots, \phi^{d-1}\},
\]
so that $K=K \Gphi\cap \GF(q)^*$ and $K\Gphi$ has size
$|K|d=etd=nt=\ord(\xi)$, hence $K\Gphi=\Gxi$.

As a consequenmce, if $L$ is a $q$-polynomial of $q$-degree $m$ over $\GF(q^m)$
that fixes $\Gphi$, then
$L(\Gxi)=L(K \Gphi)=K L(\Gphi) \subseteq K \Gphi=\Gxi$; moreover since $L$ is
one-to-one on $\GF(q)^* \{1,\phi,\ldots, \phi^{d-1}\}=\langle\gth\rangle$ and
$\Gxi\subseteq \langle\gth\rangle$, 
we conclude that in fact $L(\Gxi)=\Gxi$.
Now the desired conclusion follows from Theorem~\ref{Tqpol}.
\epf
%
%
\begin{cor}\label{Celts}
If $\phi$ is nonstandard of degree $m$ over $\GF(q)$ and if an element $\xi$ 
in some extension of $\GF(q)$ has
the same order as $\phi$, that is, if $\Gxi=\Gphi$, then $\xi$ is again
nonstandard of degree~$m$ over $\GF(q)$, with the same order and $q$-order as
$\phi$.
\end{cor}
Compare this ``extension'' results to Theorem~3.4 from \cite{bn2}.

Next, 
we present a technique to ``lift'' the nonstandardness of degree $m$
over a subfield $\bFqn$ of $\GF(q)$ to nonstandardness over $\GF(q)$, 
of the same order and sub-order, under certain conditions on $q_0$ and $q$. 
We will refer to this Theorem as the {\em lifting theorem\/}.
\begin{teor}\label{Tlift}
Let $q_0$ and $q=q_0^t$ be prime powers, and let $m$ be a positive integer for
which $(m,t)=1$. If $\xi$ is nonstandard of degree $m$ over $\bF_{q_0}$, then 
$\xi$ also is nonstandard of degree $m$ over $\GF(q)$, of the same order and 
with the $q$-order of $\xi$ equal to its $q_0$-order.
\end{teor}
\bpf
To prove the above claim, we proceed as follows. First, 
we note that $\bF_{q_0^m}\subseteq \bF_{q^i}$ with $q=q_0^t$
holds precisely when $m|ti$, hence precisely when $m|i$. So $\xi$ also has 
degree~$m$ over $\GF(q)$.

Next,
if $\xi$ has degree $m$ over $\bFqn$, then $\xi\in\bF_{q_0^m}$; hence if $\xi$ 
has order $n$, then $n|q_0^m-1$. Now according to Theorem~\ref{Tqord}, the 
$q_0$-order of $\xi$ 
is $d=n/(n,q_0-1)$ and its $q$-order is $n/(n,q-1)$. It is well-known and
easy to prove that 
\[(q_0^m-1,q_0^t-1)=q_0^{(m,t)}-1=q_0-1.\]
Hence 
\[(n,q-1)=(n,q_0^t-1)=(n,q_0^m-1,q_0^t-1)=(n,(q_0^m-1,q_0^t-1))=(n,q_0-1),\]
so that the $q_0$-order and $q$-order of $\xi$ are equal.

Now, since $(m,t)=1$, there 
is an integer $u\geq1$ such that $ut\equiv1 \bmod m$. We claim that 
\beql{Eq0q} x^{q_0} = x^{q^u}\eeql
holds for all $x\in\bF_{q_0^m}$. Indeed, since $q=q_0^t$, 
we have that (\ref{Eq0q}) holds if and only if $x^{q^u/q_0}=x^{q_0^{tu-1}}=x$ 
for all $x\in\bF_{q_0^m}$, which is the case if and only if
\[tu-1\equiv0\bmod sm,\]
which is how we have choosen $u$.
Note also that $(u,m_)=1$, hence
\beql{Eqpow} \{0,u,2u,\ldots, (m-1)u\} = \{0,1,2, \ldots, m-1\} \bmod m.\eeql

Now $\xi$ is nonstandard of degree $m$ over $\bF_{q_0}$,
so there is some nonstandard $q_0$-polynomial 
\[ L'(x)=L_0x+L_1x^{q_0}+\cdots + L_{m-1}x^{q_0^{m-1}}\]
of $q$-degree $m$ in $\bF_{q_0^m}[x]$ 
for which $L(\Gxi)=\Gxi$.
Define the ``lifted'' $q$-polynomial $L(x)\in\GF(q^m)[x]$ by
\[L(x)=L_0x+L_1x^{q^u}+\cdots + L_{m-1}x^{q^{(m-1)u}}.\]
According to (\ref{Eq0q}), we have that $L'(x)=L(x)$ on $\bF_{q_0^m}$; in
particular, we have that $L(\Gxi)=L'(\Gxi)=\Gxi$.
Moreover, obviously $L(x)$ is nonstandard if and only if $L'(x)$ is
nonstandard.
\epf
\begin{rem}\label{Rlift}
Note that if $\xi$ has the same minimal polynomial over two fields $\bK$ and
$\bL$, then by definition 
$\xi$ is nonstandard over $\bK$ if and only if $\xi$ is nonstandard
over $\bL$. 
The proof of Theorem~\ref{Tlift} can also be interpreted as showing 
that under the conditions of the
theorem, the minimal polynomial 
\[f(x)=\prod_{j=0}^{m-1}(x-\xi^{q_0^j})=\prod_{j=0}^{m-1}(x-\xi^{q^j})\]
of $\xi$ over $\bFqn$ and over $\GF(q)$ are the same.
\end{rem}
We can now use lifting and extension to construct nonstandard elements of
degree $m$ over $\GF(q)$ with $q$-order $d$, where $q=q_0^t$ with $(m,t)=1$, 
from a nonstandard element of degree $m$ over $\bFqn$ and $q_0$-order also $d$,
by applying Theorems \ref{Tlift}~and~\ref{Text}. We will use this method to
construct generalisations of Example~\ref{E1}.
\bex{E3}
Let $q_0=p^s$ with $p$ prime, let $m\geq2$ and $q_0^m>4$. Take $q=p^r$ with
$r=st$, and let $(t,m)=1$. Finally, let $\xi\in\bF_{q_0^m}$ be primitive, so
$\xi$ has order $n=q_0^m-1$. In Example~\ref{E1}, we have shown that $\xi$ is
nonstandard of degree $m$ over $\bFqn$; its $q_0$-order obviously is
$(q_0^m-1)/(q_0-1)$.
So, according to Theorem~\ref{Tlift}, $\xi$ is also nonstandard over $\GF(q)$,
of order $n=q_0^m-1$ and with $q$-order $d=(q_0^m-1)/(q_0-1)$. 

Now we can use Theorem~\ref{Text} to construct nonstandard
elements $\phi$  of $q$-order $d=(q_0^m-1)/(q_0-1)$ and order 
$N=d(q_0-1)k$ over $\GF(q)$, for all $k$ dividing $(q-1)/(q_0-1)$. All these
elements $\phi$ are powers of an element $\gth\in\GF(q^m)$ of order $d(q-1)$.
These examples all have degree $m$ over $\GF(q)$, have $q$-order
$d=(q_0^m-1)/(q_0-1)>m$ and order $n=de$ with $q_0-1|e|q-1$.

Obviously, each nonstandard element $\xi$ of degree $m$ over $\GF(q)$ with order
$n=de$ and $q$-order $d=(q_0^m-1)/(q_0-1)$ where $q_0-1|e$ and $q_0^m>4$ can be 
obtained in this way. Indeed, then $\phi=\xi^{e/(q_0-1)}$ is primitive of 
degree $m$ over $\bF_{q_0}$ and nonstandard since $q_0^m>4$, 
so $\xi$ can be obtained from the nonstandard $\phi$ by lifting and extension.

We will refer to this class of examples as {\em type II examples\/}.
\eex

%

%
%
\section{\label{Sgrm}A subgroup in $\PGL(m,q)$ related to a nonstandard
element over $\GF(q)$}
In this section, we will assume that $\xi\in\GF({q^m})$ is of degree
$m$ over $\GF(q)$, where $q=p^r$ for a prime $p$,
with order $\ord(\xi)=n$ and $q$-order $\ord_q(\xi)=d$. So 
$\eta=\xi^d\in\GF(q)^*$, and $n=de$, where $e$ is the order of $\eta$.
Furthermore, we will assume 
that $\xi$ has minimal polynomial
\[ f(x)=x^m-\gs_{m-1}x^{m-1}-\cdots -\gs_1 x-\gs_0\]
over $\GF(q)$.

Let the matrix 
\[ T=T_f=
\left(
\begin{array}{ccccccc}
0&0&0&\cdots & 0&0&\gs_0\\
1&0&0&\cdots& 0&0&\gs1\\
0&1&0&\cdots& 0& 0&\gs_2\\
\vdots&&& & & &\vdots \\
0&0&0&\cdots&0&1&\gs_{m-1}
\end{array}
\right).
\]
denote the {\em companion matrix\/} of $f$, the matrix representation of the
map $\mu: a(x)\mapsto xa(x) \bmod f(x)$ on~$\GF(q)[x] \bmod f(x)$ 
(multiplication by $x$ modulo $f(x)$
with respect to the basis $1, x, \ldots, x^{m-1}$. 
Equivalently, $T$ is the matrix representation of multiplication by $\xi$ on
$\GF(q^m)$ with respect to the basis $1,\xi,\ldots, \xi^{m-1}$ of $\GF(q^m)$,
considered as vectorspace over $\GF(q)$.
Since $f(x)|x^d-\eta$ with $\eta\in\GF(q)$ (or simply since $\xi^d=\eta$), 
we have that
\beql{ETd} T^d = \eta I.\eeql
We now first restate and prove Theorem~\ref{Tresper} from Section~\ref{Spre}.
\begin{teor}\label{Trp}
If $f(x)\in\GF(q)$ is irreducible over $\GF(q)$ and if $\xi$ is a zero of $f$,
then the restricted period $\delta(f)$ of $f$ and the $q$-order $\ord_q(\xi)$
of $\xi$ satisfy $\delta(f)=\ord_q(\xi)$.
\end{teor}
\bpf
We first note that a sequence $u=\{u_k\}_{k\geq0}$ is an $f$-sequence if and
only if the vectors
\[ u_{k,m} =(u_k, u_{k+1},\ldots, u_{k+m-1})^\top\]
satisfy 
\[u_{k+1,m}^\top=  u_{k,m}^\top T_f\]
for all $k\geq 0$.
As a consequence, if $u_{0,m}^\top=(u_0, \ldots, u_{m-1})
= (0,\ldots, 0,1)$, 
then 
$u_{0,m}^\top T^d = \gl u_{0,m}^\top$ holds
if and only if
\[u_{i,m}^\top T^d=u_{0,m}^\top T^{i+d}=\gl u_{0,m}^\top T^i = \gl u_{i,m}^\top
\]
holds for $i=0,\ldots, m-1$. Now the matrix $U$ with as its rows the vectors
$u_{i,m}^\top$ for $i=0, \ldots, m-1$ is triangular with nonzero anti-diagonal,
hence invertible. 
So from the above, we conclude that $u_{0,m}^\top T^d = \gl u_{0,m}^\top$ 
if and only if 
$UT^d=\gl U$ if and only if $T^d=\gl I$. 
\epf

From now on, we assume that, in addition, 
\[L(x)=L_0x+L_1x^q+\cdots + L_{m-1}L_{m-1}x^{q^{m-1}}\]
of $q$-degree $m$
in $\GF({q^m})[x]$ that fixes $\Gxi$, that is, there exists a permutation
$\pi\in S_n$ such that
\[ L(\xi^j)=\xi^{\pi(j)}\]
for all $j=0, \ldots, n-1$ For later use, we will also assume that $L(1)=1$.
(As we remarked earlier, this represents no loss of generality.)
For each $i$ the {\em standard\/} $q$-polynomial $L(x)=x^{q^i}$ has this
property. 
Note that according to Theorem~\ref{Tqpol}, 
$\xi$ is nonstandard over $\GF(q)$
if and only if there is a {\em nonstandard\/} $q$-polynomial as above.

By abuse of notation, we will also use $L$ to denote the $m\times m$ 
matrix representation over $\GF(q)$
of the $\GF(q)$-linear map $x \mapsto L(x)$ on $\GF(q^m)$ with respect to the
basis $1,\xi, \ldots, \xi^{m-1}$. Note that if 
\[ \xi^j = \sum_{i=0}^{m-1} c^{(j)}_i\xi^i,\]
with $c^{(j)}_i$ in $\GF(q)$ for $i=0, \ldots, m-1$ and all
$j\geq0$, then $\xi^j$ is represented by the
vector
\[ c^{(j)} = (c^{(j)}_0, \ldots, c^{(j)}_{m-1})^\top\]
in $\GF(q)^m$. As a consequence, the matrix $L$ has as its columns the vectors
$c^{(\pi(j))}$ for $j=0, \ldots, m-1$.

Let us write $\cC$ to denote the collection of all vectors $c^{(j)}$.
Then the above has the following consequence.
\begin{teor}\label{TG}
We have that
\[ T: c^{(j)} \mapsto c^{(j+1)}, \qquad L: c^{(j)} \mapsto c^{(\pi(j))},\]
so that the matrix group $G=\langle T,L\rangle$ in $\GL(m,q)$ fixes the 
collection $\cC$ as a
set. 
\end{teor}

For later use, we also consider the following ``normalisation''.
Write $\gs=\gs_{m-1}=\Tr_{\GF(q^m)/\GF(q)}$, 
and assume that $\gs\neq0$. Let $\txi=\xi/\gs$. Then
$\txi\in\GF(q^m)$ again has $q$-order~$d$ and degree $m$ over
$\GF(q)$, with minimal polynomial 
\[ \tf(x)=x^m-\tgs_{m-1}x^{m-1}-\cdots -\tgs_1 x-\tgs_0,\]
where $\tgs_i=\gs_i/\gs^{m-i}$; in particular, $\tgs_{m-1}=1$. 
So $1,\txi, \ldots, \txi^{m-1}$ are another basis for $\GF(q^m)$ over $\GF(q)$.
We now write
\[ \txi^j = \sum_{i=0}^{m-1} \tce^{(j)}_i\txi^i,\]
with $\tce^{(j)}_i$ in $\GF(q)$ for $i=0, \ldots, m-1$ and
all $j\geq0$, so that $\txi^j$ is represented by the
vector
\[ \tce^{(j)} = (\tce^{(j)}_0, \ldots, \tce^{(j)}_{m-1})^\top.\]
Note that 
\[ \tce^{(j)}_i = \gs^{-j+i}c^{(j)}_i\]
for all $j\geq0$ and all $i=0,\ldots, m-1$.

The {\em conjugate matrix\/} $M^S$ of a matrix $M$ by an invertible matrix~$S$ 
is defined as $M^S=SMS^{-1}$. Note that the conjugate $M^S$ is the matrix
representation of the same linear map, but with respect to a basis
transformation given by $S$.
We will write $\tT$ to denote the companion matrix of $\tf(x)$.
Define the diagonal matrix $D$ as
\[ D={\rm diag}(1,\gs, \ldots, \gs^{m-1}).\]
Our observations are summarized in the following theorem.
\begin{teor}\label{Tnorm}
With the above definitions, we have that $Dc^{(j)}=\gs^j \tce^{(j)}$.
Moreover,  the conjugate $T^D=DTD^{-1}$ of $T$ satisfies $T^D=\gs \tT$, and
$T^D$ and the conjugate $L^D=DLD^{-1}$ of $L$ satisfy
\[ T^D : \tce^{(j)}\mapsto \gs\tce^{(j+1)}, \qquad 
L^D: \tce^{(j)}\mapsto \gs^{\pi(j)-j} \tce^{(\pi(j)}.\]
So the conjugate group $G^D=\langle T^D, L^D\rangle$ fixes the set
$\tcC=\{\tce^{(j)}\mid j=0, \ldots, n-1\}$ as a set.
\end{teor}

In the remainder of this paper, we will use the groups $G$ and $G^D$ to obtain
information on $\xi$ and $L$, and, in particular, on the $q$-order
$\ord_q(\xi)$ of the nonstandard element $\xi$.
To this end, we will consider the sets $\cC$ and $\tcC$ as subsets of
$\PG(m-1,q)$, and the groups $G$ and $G^D$ as subgroups of $\PGL(m,q)$, in its
natural action on $\PG(m-1,q)$. Here, $\PG(m-1,q)$ consists of the lines
through the origin in $\GF(q)^n$. Equivalently, $\PG(m-1,q)$ consists of the
nonzero vectors $v$ from $\GF(q)^n$, where we identify a vector $v$ with its
scalar multiples $\gl v$ for $\gl \in \GF(q)^*$. The group $\PGL(m,q)$
consists of the collection $\GL(m,q)$ of all nonsingular $m\times m$ matrices
over $\GF(q)$, where we identify a matrix $M$ with its scalar multiples $\gl
M$, for $\gl \in \GF(q)^*$.

Now we assumed that $\xi$ has $q$-order $d$, with $\xi^d=\eta\in\GF(q)$, 
so we have that the vector $c^{(d)}$ representing $\xi^d$ satisfies
$c^{(d)}=\eta c^{(0)}$. Since $c^{(j)}=T^jc^{(0)}$ for all $j$, we see that the
set $\cC$, considered as {\em subset of $\PG(m-1,q)$\/}, has size $d=\ord_q(\xi)$.
Note furthermore that since $T^d=\eta I$, the matrix $T$ has order $d$ as
element of the group $\PGL(m,q)$.
Note also that
\[ \Gxi=\langle\eta\rangle \cup \langle\eta\rangle \xi\cup \ldots
\cup \langle\eta\rangle\xi^{d-1},\]
where the union is {\em disjoint\/}.
As a consequence, there exists a permutation $\tau\in S_d$ such that
\[ L(\xi^k)=\eta_k\xi^{\tau(k)}\]
with $\eta_k\in\langle\eta\rangle\subseteq\GF(q)^*$, for all $k=0, \ldots, n-1$.
So $L$, as an element of $\PGL(m,q)$, acts on $\cC$, considered as a subset 
of $\PG(m-1,q)$, by 
\[ L: c^{(j)} \mapsto c^{(\tau(j)},\]
for $j=0, \ldots, d-1$.
We summarize the above in the next theorem.
\begin{teor}\label{Tprojm} The groups $G$ and $\G^D$ obtained from a
nonstandard element $\xi$, considered as subgroups
of $\PG(m,q)$, have orbits $\cO=\{c^{(j)}\mid j=0, \ldots, d-1\}$ and
$\cO^D=\{\tce^{(j)}\mid j=0, \ldots, d-1\}$, respectively. Both $\cO$ and
$\cO^D$ have size $d=\ord_q(\xi)$ and contain $1=(1,0,\ldots, 0)^\top$.
\end{teor}
%
%
\section{\label{Sgr2}The case $m=2$}
%
We now investigate the case where $m=2$ in more detail. 
{\bf So from now on, we will assume that $m=2$.}

So here $\xi\in\GF(q^2)\setminus\GF(q)$ is zero of the irreducible polynomial
$f(x)=x^2-\gs_1 x-\gs_0$ over $\GF(q)$, where we assume that $\gs_1\neq 0$. 
(So we assume that $d=\ord_q(\xi)>2$.)
Writing $\gs=\gs_1$ and $\gl = \gs_0/\gs_1^2$, we also have that $\txi=\xi/\gs$ 
is zero of the polynomial $\tf = x^2-x-\gl$. Note that, as a consequence, we
have that
\beql{Etxi} \txi^q = 1-\txi.\eeql

Again, we assume that the $q$-polynomial $L(x)=L_0x+L_1 x^q$ of $q$-degree $2$
over $\GF(q^2)$ fixes $\Gxi$ as a set.
As remarked before, we may assume without loss of generality that $L(1)=1$. 
Let $\gw, \nu\in\GF(q)$ be such that
\[L(1)=1, \qquad L(\xi)=\gw+\nu \xi.\]
Put $\tgw=\gw/\gs$. Then 
\[ L(1)=1, \qquad L(\txi)=\tgw + \nu \txi,\]
so that the matrix representations $L$ and $L^D$ of the map induced by the
polynomial $L$ on $\GF(q^2)$ are given by
\[
L=
\left(
\begin{array}{cc}
1&\gw\\
0&\nu
\end{array}
\right), 
\qquad
L^D=
\left(
\begin{array}{cc}
1&\tgw\\
0&\nu
\end{array}
\right). 
\]
Finally, it is easily verified that the matrices $T$ (multiplication by $\xi$) 
and $\tT=\gs^{-1} T^D$ (multiplication by $\txi$) are given by 
\[
T=
\left(
        \begin{array}{cc}
                0 & \gs_0 \\
                1 & \gs_1
        \end{array}
        \right), 
\qquad
T^D=\gs
\left(
        \begin{array}{cc}
                0 & \gl \\
                1 & 1
        \end{array}
        \right). 
\]
In what follows, we will investigate the subgroup $\Xi=\langle \Lambda,
\Gamma\rangle$ of $\PGL(2,q)$ generated by the elements
\beql{ELG}
\Lambda=
\left(
        \begin{array}{cc}
                0 & \gl \\
                1 & 1
        \end{array}
        \right), 
\qquad 
\Gamma = 
\left(
\begin{array}{cc}
1&\tgw\\
0&\nu
\end{array}
\right).
\eeql
We will employ the usual identification of $\PG(1,q)$ with the
set $\GF(q)\cup \{\infty\}$ by identifying the element $(x,y)\in\PG(1,q)$ with 
the finite field element $x/y\in\GF(q)$ if $y\neq0$ and with $\infty$ if $y=0$.
As a consequence, a matrix 
\[
M=
\left(
        \begin{array}{cc}
                a & b \\
                c & d
        \end{array}
        \right)
\]
from $\PGL(2,q)$ now acts on an element $x$ from 
$\GF(q)^+=\GF(q)\cup\{\infty\}$ as
\[ M: x \mapsto (ax+b)/(cx+d).\]
So now the field element $\xi^j\in\GF(q^2)$ corresponds to 
$c^{(j)}=(c^{(j)}_0,c^{(j)}_1)\sim c^{(j)}_0/c^{(j)}_1$ in $\GF(q)^+$; 
in particular, we have that $1=c^{(0)}\sim \infty$ and $\xi=c^{(1)}\sim 0$.
In the next theorem, we summarize the main consequences of the above
definitions and assumptions.
\begin{teor}\label{Tgr2}
Let $\xi$ have degree $m$ over $\GF(q)$, with $q$-order $d=\ord_q(\xi)>2$ and
minimal polynomial $f(x)=x^2-\gs x -\lambda \gs^2$. Let
$L$ be a $q$-polynomial of $q$-degree $m$ over $\GF(q^m)$ that fixes $\Gxi$, 
with $L(1)=1$ and $L(\xi)=\gs \tgw +\nu \xi$,
and let $\Lambda$ and $\Gamma$ be the associated matrices as in (\ref{ELG}).
Then the following holds.\\
(i)
The element $\Lambda$ has order $d$ in $\PGL(2,q)$, and no fixed points on
$\GF(q)^+$. 
Moreover, we have that
\[ \Lambda^k = \left(
\begin{array}{cc}
\gl F_{k-1}&\gl F_k\\
F_k&F_{k+1}
\end{array}
\right) ,
\]
where the $F_k$ are defined by $F_0=0$, $F_1=1$, and $F_{k+2}=F_{k+1}+\gl F_k$
for all $k$. In particular, $\Lambda$ has order $d$ and $F_k=0$ if and only if
$k\equiv 0\bmod d$.
The matrix $\Lambda$ induces a map $x\mapsto \gl/(1+x)$ on $\GF(q)^+$.\\
(ii)
The element $\Gamma$ induces a map $x\mapsto (x+\tgw)/\nu$ on $\GF(q)^+$. 
We have that 
\[ \Gamma^k = \left(
\begin{array}{cc}
1&\tgw(1+\nu+\cdots+\nu^{k-1})\\
0&\nu^k
\end{array}
\right) .
\]
In particular, if $\nu$ has order $e$, then $\Gamma$ has order $e$ (if
$\nu\neq1$ or $\tgw=0$) or $p$ (if $\nu=1$ and $\tgw\neq0$).\\
(iii) The subset 
$\cO=\{\infty, \Lambda(\infty), \ldots, \Lambda^{d-1}(\infty)\}$ of $\GF(q)^+$
is an orbit of the subgroup $\Xi$ of $\PGL(2,q)$
generated by the maps $\Lambda$ and $\Gamma$.
The ``standard'' $q$-polynomials $L(x)=x$ or $L(x)=x^q$ correspond to the cases
$\nu=1$, $\gw=0$, $\tgw=0$, and 
$\nu=-1$, $\gw=\gs_1$,
$\tgw=1$, respectively.
\end{teor}
\bpf
Most of the claims are a direct consequences of our assumptions and
definitions. 
Hence all orbits of $\Lambda$ on $\GF(q)^+$ have the same size $d$.
The claim concerning the case where $L(x)=x$ is evident.
Finally, since $\xi$ and $\xi^q$ are the zeroes of $f(x)=x^2-\gs_1 x-\gs_0$, 
we have that
\[ \gs_1=\xi+\xi^q, \qquad \gs_0=-\xi^{q+1}.\]
Hence if $L(x)=x^q$, then $\nu\xi+\omega=L(\xi)=\xi^q=\gs_1-\xi$, so that
$\nu=-1$ and $\omega=\gs_1$.
\epf


It turns out that the cases $d=3,4,5$ need a special treatment. For later use,
we now collect the required extra information.
Note that according to Theorem~\ref{Tgr2}, 
the orbit $\cO$ has size $d$ and is given by
\beql{Eorbit}\cO=\{\infty, 0, \gl, \gl/(1+\gl), \gl(1+\gl)/(1+2\gl),
\gl(1+2\gl)/(1+3\gl+\gl^2), \ldots\}.
\eeql
\begin{lem}\label{Lnod3}
There are no nonstandard $\xi$ of degree 2 over $\GF(q)$ with $q$-order $d=3$.
\end{lem}
\bpf
From (\ref{Eorbit}) we see that if $d=3$, then necessarily $\gl=-1$. 
Now since $\cO$ is also invariant under $\Gamma$, we have that
\[ \cO = \{\infty, 0, -1\}=\{\infty, \tgw/\nu, (-1+\tgw)/\nu\}\].
So we have one of two cases:
\begin{enumerate}
\item
$\tgw=0$. Then $\nu=1$, so we are in the case where $L(x)=x$.
\item
$\tgw=1$. Then $\nu=-1$, so we are in the case where $L(x)=x^q$.
%
\end{enumerate}
Since there are no other possibilities, the claim follows.
\epf
\begin{lem}\label{Ld4}
If $\xi$ is nonstandard of degree 2 over $\GF(q)$ with $q$-order $d=4$, 
then $p=3$ and $\txi=\xi/\gs$ is primitive in $\GF(9)$.
Moreover, $\Xi$ is actually a subgroup of $\PGL(2,3)$.
\end{lem}
\bpf
From (\ref{Eorbit}) we see that if $d=4$, then necessarily $\gl=-1/2$. 
Now since $\cO$ is also invariant under $\Gamma$, we have that
\[ \cO = \{\infty, 0, -1/2, -1\}=\{\infty, \tgw/\nu, (-1/2+\tgw)/\nu,
(-1+\tgw)/\nu\}\].
So we have one of three cases.
\begin{enumerate}
\item
$\tgw=0$. Then $\{-1/2,-1\}=\{(-1/2)/\nu, -1/\nu\}$, so either
$\nu=1$ (which corresponds to the case where $L(x)=x$), or $\nu=1/2=2$, 
so the characteristic $p=3$ and $\nu=-1$. 
\item
$\tgw=1/2$. Then $\{-1/2,-1\}=\{(1/2)/\nu, (-1/2)/nu\}$, so
$p=3$, $\tgw=1/2=-1$, and either $\nu=-1$ or $\nu=1$.
\item
$\tgw=1$. Then $\{-1/2,-1\}=\{1/\nu,(1/2)/\nu\}$, so either
$\nu=-1$ (which corresponds to the case where $L(x)=x^q$), or $\nu=-1/2=-2$, so
that $p=3$, $\nu=1$, and $\tgw=1$.
\end{enumerate}
We are left with four cases. All have $p=3$, so that
$\txi$ is zero of $x^2-x-\gl=x^2-x-1$ and $\txi$ is
primitive in $\GF(9)$. 
Note that these four remaining cases represent 
the different nonstandard ways of mapping the nonstandard subgroup
$\GF(9)^*$ onto itself.
In all these cases, $\gl$, $\tgw$, and $\nu$ are in $\GF(3)$, hence $\Xi$ is
actually a subgroup of $\PGL(2,3)$.
%
\epf
\begin{lem}\label{Ld5}
If $\xi$ is nonstandard of degree 2 over $\GF(q)$ with $q$-order $d=5$, 
then $p=2$, and $\txi=\xi/\gs$ is primitive in $\GF(16)$.
Moreover, $\Xi$ is actually a subgroup of $\PGL(2,4)$.
\end{lem}
\bpf
From (\ref{Eorbit}) we see that if $d=5$, then necessarily 
$\gl^2+3\gl+1=0$, so that $\gl=-(\gl+1)^2$. 
For later use, we remark that in characteristic $p=2$, we have that $\gl$ is
primitive in $\GF(4)$ and $\txi$, the zero of 
$x^2+x+\gl$, is primitive in $\GF(16)$. 
So we are done if we can prove that in all cases $p=2$, or $(\nu,\tgw)=(1,0)$
(corresponding to the case where $L(x)=x$), or $(\nu,\tgw)=(-1,1)$
(corresponding to the case where $L(x)=x^q$).

Now since $\cO$ is also invariant under $\Gamma$, we have that
\[ \cO = \{\infty, 0, \gl, -1-\gl, -1\}=
\{\infty, \tgw/\nu, (\gl+\tgw)/\nu, (-1-\gl+\tgw)/\nu, (-1+\tgw)/\nu\}.\]
So we have one of four cases for $\tgw$.
\begin{enumerate}
\item
$\tilde{\omega}=0$ and 
$\{\gl, -1-\gl,-1\}=\{\gl/\nu, (-1-\gl)/\nu, -1/\nu\}$.
Then $\nu\in\{1, \gl/(-(1+\gl)), -\gl\}= \{1, 1+\gl, -\gl\}$, so we have one of
the following.
\begin{itemize}
\item[\rm (a)]
$\nu=1$ (corresponding to the case where $L(x)=x$);
\item[\rm (b)]
$\nu=1+\gl$, $\{\gl,-1\}=\{-1,-1/(1+\gl)\}$. Hence $\gl=-1/(1+\gl)$, or
$\gl^2+\gl+1=0$; combined with the other equation for $\gl$ this shows that, in
addition, $p=2$.
\item[\rm (c)]
$\nu=-\gl$, $\{\gl,-1-\gl\}=\{(1+\gl)/\gl,1/\gl\}$. So either
$\gl=(1+\gl)/\gl$ and $-1-\gl=1/\gl$, whence $p=2$, or $\gl=1/\gl$, which leads
to an impossibility.
\end{itemize}
\item
$\tilde{\omega}=-\gl$ and
$\{\gl, -1-\gl,-1\}=\{-\gl/\nu, (-2\gl-1)/\nu, (-\gl-1)/\nu\} =
\{-\gl/\nu, (\gl^2+\gl)/\nu, (-\gl-1)/\nu\}$. 
Then $\nu\in\{-1, \gl/(1+\gl)=-(1+\gl), \gl\}$, so we have one of the
following.
\begin{itemize}
\item[\rm (a)] $\nu=-1$. This leads to $p=2$ or an impossibility.
\item[\rm (b)] $\nu=-(1+\gl)$. Then either $p=2$, or an impossibility.
\item[\rm (c)] $\nu=\gl$. Here either $p=2$ or an impossibility.
\end{itemize}
\item
$\tilde{\omega}=1+\gl$ and
$\{\gl, -1-\gl,-1\}=\{(1+\gl)/nu, (2\gl+1)/\nu,\gl/\nu\}=
\{(1+\gl)/\nu, -(\gl^2+\gl)/\nu,\gl/\nu\}$, so we have one of the following.
\begin{itemize}
\item[\rm (a)] $\nu=-(1+\gl)$. Then we have $-(1+\gl)=-\gl/(1+\gl)$, 
which leads to $p=2$. 
\item[\rm (b)] $\nu=-(\gl^2+\gl)$. Then either an impossiblity or 
$\gl=-1/(1+\gl)$, which leads to $p=2$.
\item[\rm (c)] $\nu=-\gl$. Then either $p=2$, or an impossibility.
\end{itemize}
\item
$\tilde{\omega}=1$ and
$\{\gl, -1-\gl,-1\}=\{1/\nu, (1+\gl)/\nu, -\gl/\nu\}$, so we have one of the
following. 
\begin{itemize}
\item[\rm (a)] $\nu=-1$. This corresponds to the case where $L(x)=x^q$.
\item[\rm (b)] $\nu=-(1+\gl)$. This leads to $p=2$ or an impossibility.
\item[\rm (c)] $\nu=\gl$. Here either $\gl=1/\gl$ and $-1-\gl=(1+\gl)/\gl$, 
which leads to $p=2$, or
$\gl=(1+\gl)/\gl$ and $-1-\gl=1/\gl$, which again leads to $p=2$.
\end{itemize}
\end{enumerate}
So in all nonstandard cases we have $p=2$.
Moreover, in all these cases, $\gl$, $\tgw$, and $\nu$ are in $\GF(4)$, hence 
$\Xi$ is actually a subgroup of $\PGL(2,4)$.
\epf
\section{\label{Sgroup}A subgroup in $\PGL(2,q)$}
The groups $\PGL(2,q)$ are one of the few groups for which the
complete subgroup structure is known. 
In this section, we will use this knowledge to obtain further information on 
the subgroup $\Xi$ of $\PGL(2,q)$ from Theorem~\ref{Tgr2}.
For our purposes, the following is sufficient.
\begin{teor}[\cite{hup}, \cite{suz}, \cite{dick}]\label{Tpgl-sub}
Let $q=p^r$ with $p$ prime. 

\noindent
(i) If $M$ is a non-identity element in $\PGL(2,q)$ of order $k$, 
with $f$ fixed points, then all orbits of size $>1$ have
size $k$, and either $f=1$, $k=p$, or $f=2$, $k|q-1$, or $f=0$, $k|q+1$.

\noindent
(ii) The subgroups of $\PGL(2,q)$ are as follows:
\begin{enumerate}
\item Cyclic subgroups $C_k$, of order $k=2$ (if $p$ is odd), or of order $k>2$
with $k|q\pm1$.
\item Dihedral subgroups $D_{2k}$ of order $2k$, with $k=2$ (if $p$ is odd), or
with $k>2$ and $k|q\pm1$.
\item Elementary abelian subgroups $E_{p^k}$, of order $p^k$ with $0\leq k\leq
r$.
\item A semidirect product of the elementary subgroup $E_{p^k}$, where $1\leq
k\leq r$, and the cyclic group $C_\ell$, where $\ell|q-1$ and $\ell|p^k-1$.
\item Subgroups isomorphic to $A_4\cong\PSL(2,3)$, $S_4\cong\PGL(2,3)$, or 
$A_5\cong\PSL(2,4)$.
\item One conjugacy class of subgroups isomorphic to $\PSL(2,p^k)$, where $k|r$.
\item One conjugacy class of subgroups isomorphic to $\PGL(2,p^k)$, where $k|r$.
\end{enumerate}
\end{teor}

In the references, the classifications are given for subgroups of $\PSL(2,q)$.
If $q$ is even, then $\PSL(2,q)=\PGL(2,q)$.
To obtain the classification for $\PGL(2,q)$, note that
if $q$ is odd, then
$\PGL(2,q)$ is a subgroup of $\PSL(2,q^2)$ and has a unique subgroup 
$\PSL(2,q)$, of index two.
A similar classification
has been used e.g.\ in \cite{cot} and \cite{cmot} in the case where $q$ is odd.

We now use this classification to show the following.
\begin{teor}\label{TXi}
The group $\Xi$ from Theorem~\ref{Tgr2} is one of the following.
\begin{itemize}
\item
A cyclic group, in the case where $L(x)=x$;
\item
a dihedral group, in the case where $L(x)=x^q$;
\item
a group of the form $\PSL(2,q_0)$ or $\PGL(2,q_0)$,
in the nonstandard case,
with $d=q_0+1>3$ and $q=q_0^t$, where $t$ is odd.
\end{itemize}
\end{teor}
\begin{proof}
We break the proof into a number of cases.

\noindent
(1) 
The subgroup $\Xi$ cannot be cyclic except when $L(x)=x$. 
Indeed, if $\Xi$ is cyclic, then $\Lambda$ and $\Gamma$ commute, that is,
$\Lambda\Gamma=\Gamma\Lambda$. It is easily verified that this happens if and
only if $\tgw=0$ and $\nu=1$, that is, if $\Gamma=I$.

\noindent
(2) 
The subgroup $\Xi$ cannot be dihedral except when $L(x)=x^q$.
Indeed, $\Lambda$ has order $d>2$, so if $\Xi$ is dihedral, then both $\Gamma$
and $\Gamma\Lambda$ have order two. Hence $\nu=-1$ and $\tgw=-\nu$, so according
to 
Theorem~\ref{Tgr2}, we have $L(x)=x^q$.

\noindent
(3) 
The subgroup $\Xi$ cannot be elementary abelian of order $p^k$. 
Indeed, since $d|q+1$, we have $(d,p)=1$, so the order of $\Lambda$ cannot be a
power of $p$.

\noindent
(4)
The subgroup $\Xi$ cannot be semisimple product of an elementary abelian group 
of order $p^k$ with a cyclic group of order $\ell$, if $\ell|q-1$ and
$\ell|p^k-1$. Indeed, suppose that this would be the case.
The semidirect product has cardinality $p^k\ell$, and
since $(d,p)=1$, we would conclude that $d|\ell$, hence $d|q-1$. Now we also 
have that $d|q+1$, so it would follow that $d|2$, which is impossible if $d>2$.

\noindent
(5) 
If the subgroup $\Xi$ is one of $A_4$, $S_4$, or
$A_5$, then the order $d>2$ 
of the element $\Lambda\in \Xi$ is one of $3$, $4$, or $5$. These cases were
handled in Section~\ref{Sgr2}. In Lemma~\ref{Lnod3}, it was shown that the case
$d=3$ is not possible. In Lemma~\ref{Ld4} it was shown that if $d=4$, then 
$p=3$ and the group $\Xi$ is a subgroup of $\PGL(2,q_0)$ for $q_0=3$.
Finally, in Lemma~\ref{Ld5}, it was shown that if $d=5$, then $p=2$ and the
group $\Xi$ is a subgroup of $\PGL(2,q_0)$ with $q_0=4$.
As a consequence, since we are not in one of the cases  (1-4) above, we must
have one of the cases (6), (7) below.

\noindent
(6), (7) 
Here we have that $\Xi$ is isomorphic to either $\PSL(2,q_0)$ or $\PGL(2,q_0)$, 
with $\GF(q_0)$ a subfield of $\GF(q)$.
Such a subgroup is conjugated in $\PGL(2,q)$ to the ``obvious'' subgroup
consisting of invertible matrices with entries in $\bFqn$. 
It is easily verified that these two groups both have
one orbit $\GF(q_0)^+$ of size $q_0+1$, and one orbit
$\GF(q_0^2)\setminus\GF(q_0)$ of size $q_0^2-q_0$. Moreover, it is easy to show
that both groups act regularly on $\GF(q^2)\setminus \GF(q_0^2)$.
Indeed, this immediately follows from the fact that the fixed points in
$\GF(q)^+$ of a non-identity map $M: x\mapsto (ax+b)/(cx+d)$ with 
$a,b,c,d\in\GF(q_0)$ are the zeros of the non-trivial polynomial $cx^2
+(d-a)x-b$ of degree two over $\GF(q_0)$, hence are contained in $\GF(q_0^2)$.
(For more details, see, e.g., \cite{hq}.)
So all other orbits of $\Xi$ are of size $|\Xi|$, 
hence have a size equal to $q_0(q_0^2-1)$ or $q_0(q_0^2-1)/2$ (if $q$ is odd and
$G_q=\PSL(2,q_0)$. Since $\Xi$ has an orbit of size $d$ and since $d|q+1$, we
have $(d,q)=1$ and hence we
must have that $d=q_0+1$.

Now note that since $\Lambda$ has order $d$ and no fixed points, all
orbits of $\Lambda$ have size $d$, hence $d|q+1$.
So in the nonstandard case, we have $q=p^r$ and $d=q_0+1$, with $q_0$ of the
form $q_0=p^s$ and with $s|r$, that is, with $q=q_0^t$ for some $t$. 
Since $d|q+1$, it follows that $t$ is odd. 
\end{proof}
Next, we want to show that if $G$ is isomorphic to $\PSL(2,q_0)$ or
$\PGL(2,q_0)$, then $\gl$, $\nu$, and $\tgw$ are actually contained in
$\bFqn$.
To this end, we need some preparation. 
If $M=(M_{i,j})$ is a matrix over $\GF(q)$, where $q=p^r$, then we write 
$M^{(p^s)}$ to denote the matrix with entries $M_{i,j}^{p^s}$.
\begin{lem}[\cite{hq}]\label{Lhq1}
Let $q=q_0^t$ and let $M\in\PGL(2,q)$. Then $M$ is
contained in $\PGL(2,q_0)$ if and only if $M^{(q_0)}=\phi M$ holds for some
$\phi\in\GF(q)^*$.
\end{lem}
\bpf
The proof in \cite{hq} uses Galois theory. For completeness' sake, we sketch a
simple proof here. (In fact, many different proofs are possible.)
The matrix
\[ M = \left(
\begin{array}{cc}
a & b\\
c & d
\end{array}
\right) 
\]
is in $\PGL(2,q_0)$ precisely when some multiple $\gb M$ of $M$ has all its
entries in $\bFqn$, so when the map $x\mapsto (ax+b)/(cx+d)$ fixes
$\bFqn^+$ as a set. Now $((ax+b)/(cx+d))^{q_0}=(ax+b)/(cx+d)$ for all
$x\in\bFqn^+$ leads to a second degree equation that is identically zero on
$\bFqn$, so has all coefficients equal to zero. From the resulting three
equations, the lemma follows.
\epf
Next, for a matrix $M$ over $\GF(q)$, we write $\det(M)$ and $\Tr(M)$ to denote
the determinant and trace of $M$, respectively. Also, we write $M^A$ to denote
the conjugate $AMA^{-1}$ of $M$ by $A$.
over 
\begin{lem}\label{Lhq2} Let $\bFqn$ be a subfield of $\GF(q)$.

(i) A matrix $M$ over $\GF(q)$ is contained in a subgroup of
$\PGL(2,q)$ isomorphic to $\PSL(2,q_0)$ or $\PGL(2,q_0)$ if and only if 
$(M^A)^{(q_0)} = \phi M^A$ for some matrix $A$ in $\PGL(2,q)$ and some
$\phi\in\GF(q)^*$, where
$\phi^2=\det(M)^{q_0-1}$ and either $\phi=\Tr(M)^{q_0-1}$ or $\Tr(M)=0$.

(ii) If a matrix
$M$ over $\GF(q)$ is contained in a subgroup of $\PGL(2,q)$ isomorphic to
$\PSL(2,q_0)$ or $\PGL(2,q_0)$, then either $\Tr(M)=0$ or 
$\Tr(M)^{2(q_0-1)}= \det(M)^{q_0-1}$.
\end{lem}
\bpf
If $M$ is contained in some subgroup of
$\PGL(2,q)$ isomorphic to $\PSL(2,q_0)$ or $\PGL(2,q_0)$, then
there is a matrix $A\in\PGL(2,q)$ such that $M^A$ is contained in $\PSL(2,q_0)$
or $\PGL(2,q_0)$, that is, according to Lemma~\ref{Lhq1},
\beql{EFixA} (M^A)^{(q_0)} = \phi M^A, \eeql 
for some $A\in\PGL(2,q)$ and some $\phi\in\GF(q)^*$.
Now $\det(M^A) = \det(A)$ and $\Tr(X^A) = \Tr(X)$, so
from (\ref{EFixA}), we conclude that
\[ \det(M)^{q_0} = \phi^2 \det(M), \qquad \Tr(M)^{q_0} = \phi \Tr(M), \]
hence $\phi^2=\det(M)^{q_0-1}$ and either $\Tr(M)=0$ or $\phi=\Tr(M)^{q_0-1}$.
\epf

Now we apply this result to our matrices $\Lambda$ and $\Gamma$.
The result is as follows.
\begin{teor}\label{TXi-sub}
In the nonstandard case, there exists a prime power $q_0$ such that $d=q_0+1>3$,
$q=q_0^t$ with $t$ odd, and the subgroup 
$\Xi=\langle \Lambda, \Gamma\rangle$ of $\PGL(2,q)$ generated by
$\Lambda$ and $\Gamma$ as in Theorem~\ref{Tgr2} 
is equal to either $PSL(2,q_0)$ or $\PGL(2,q_0)$.
Moreover, we have 
$\gl, \nu, \tgw \in \bFqn$, 
and $\bFqn$ is the smallest subfield of $\GF(q)$
containing $\gl$.
\end{teor}
\bpf Acccording to Theorem~\ref{TXi}, in the nonstandard case we have $q=q_0^t$
with $t$ odd, $d=q_0+1$, and $\Xi$ conjugate in $\PGL(2,q)$ to either 
$\PSL(2,q_0)$ or $\PGL(2,q_0)$.
Now first, since $\det(\Lambda)=-\lambda$ and $\Tr(\Lambda)=1$, we see from
Lemma~\ref{Lhq2} that 
$\lambda$ must be contained in $\bFqn$.
Next, since the orbit $\langle\Lambda\rangle(\infty)$ of $\Lambda$ 
containing $\infty$
has size $d=q_0+1$, it must be equal to $\bFqn^+$; since it is fixed by $\Xi$, 
we must now have $\Gamma(\bFqn^+)=\bFqn^+$.
This immediately implies that both $\nu$ and $\tgw$ must be contained 
in~$\bFqn$.
\epf
Wev will now use this result to show the following.

\begin{teor}\label{Tlift-ext}
A nonstandard elements of degree two over a field $\GF(q)$ with $q$-order
$d$ is either of type II, with $d=2$ and of the form as in 
Example~\ref{E2}, so has $n=2e$
with both $q$ and $(q-1)/e$ odd, or is of type I and has $d\geq4$
of the form $d=q_0+1$, for some $q_0$ such that $q=q_0^t$ with $t$ odd, and can
be obtained from a nonstandard element of degree two over $\bFqn$ with
$q_0$-order $q_0+1$ by lifting and extension as in Theorems
\ref{Tlift}~and~\ref{Text}. 
\end{teor}
\bpf
Let $\xi$ be nonstandard of degree $m=2$ over $\GF(q)$, with minimal polynomial 
$f(x)=x^2-\gs x -\gs^2 \gl$ over $\GF(q)$, and let $\xi$ has
order $n=de$ and $q$-order $d=\ord_q(\xi)$. According to Theorem~\ref{Tqord},
we have $e=(n,q-1)|q-1$ and $(d,(q-1)/e)=1$.

Now $d=2$ if and only if $\gs=0$; in that case $\xi$ is of type II, so as in
Example~\ref{E2}.
So in addition we will assume that $d>2$ and $\gs\neq0$. Write $\txi=\xi/\gs$. 
Then $\txi$ has minimal polynomial $\tf=x^2-x-\gl$. Let $L(x)$ be a nonstandard
$q$-polynomial of $q$-degree two over $\GF(q^2)$ that fixes $\Gxi$, with
$L(1)=1$ and $L(\xi)=\gw+\nu\xi$. Write $\tgw=\gw/\gs$. 
Then $L(\txi)=\tgw+nu\txi$.

According to Theorem~\ref{TXi-sub}, we now have that $d=q_0+1\geq4$, 
where $q=q_0^t$ with $t$ odd, and $\gl,\tgw,\nu\in\bFqn$, with
$(\nu,\tgw)\neq(1,0), (-1,1)$.
We claim that $L$ is a $\bFqn$-linear map of $q_0$-degree two on $\bFqqn$.
This can be shown as in the proof of the ``lifting'' theorem, but can also
shown directly, as follows. Since $\txi$ has minimal polynomial
$\tf(x)=x^2-x-\gl$ with $\gl\in\bFqn$, we have that $1,\txi$ is a basis for
$\bFqqn$ over $\bFqn$; moreover, since $t$ is odd, we have
$\txi^q=\txi^{q_0^t}=\txi^{q_0}$. Hence $L$ is $q_0$-linear over $\bFqqn$.

Now $L$ is a bijection on $\Gxi$ and maps $\bFqqn$ into $\bFqqn$; we conclude
that $L$ is also a bijection on $\Gxi\cap\bFqqn=\Gphi$, where
$\phi=\xi^{\gd_0}$ with $\gd_0$ the $q_0^2$-order of $\xi$.
Hence $\phi$ is nonstandard of degree 2, both over $\GF(q)$ and over $\bFqn$.

For later use, we want to show that $(q-1)/e$ must be odd.
Indeed, we have $n=de$ with $d=q_0+1$ and $e|q-1$. Now, as easily verified,
\[(q_0+1,q-1)=(q_0+1,q_0^t-1)=
\left\{
\begin{array}{cc}
        1, & \mbox{if $q_0$ is even};\\
        2, & \mbox{if $q_0$ is odd},
\end{array}
\right.
\]
hence
\[q_0+1=d=n/(n,q-1)=(q_0+1)e/((q_0+1)e,q-1) = (q_0+1)/(q_0+1,(q-1)/e)\]
holds precisely when $(q-1)/e$ is odd.

Next, by Theorem~\ref{Tqord}, we have that the $q_0^2$-order $\gd_0$ of $\xi$
is given by
\[\gd_0=n/(n,q_0^2-1)=e/(e,q_0-1)=e/e_0,\]
where $e_0=(e,q_0-1)$. Note that $\phi=\xi^{\gd_0}$ has order 
$n_0=n/\gd_0=(q_0+1)e_0$. We claim that the $q_0$-order $d_0$ of $\phi$ is
equal to $d$. Indeed, 
\[d_0=n_0/(n_0,q_0-1)=(q_0+1)e_0/((q_0+1)e_0,
q_0-1)=(q_0+1)/(q_0+1,(q_0-1)/e_0,\]
so we are done if $(q_0-1)/e_0=(q_0-1)/(e,q_0-1)$ is odd, which follows
immediately from the fact that $q_0-1|q-1$ (so $e$ contains every
factor 2 contained in $q-1$, so certainly all factors 2 contained in $q_0-1$).

Finally, we want to show that $\xi$ can be obtained from $\phi$ by lifting and
extension. Now lifting shows that, as remarked earlier, 
$\phi$ is also nonstandard of degree 2 over $\GF(q)$. We know by definition of
$\phi$ that $\Gphi\subseteq\Gxi$, hence according to Theorem~\ref{Text}
we only have to show that $\xi\in\GF(q)^*\Gphi$. To this end, write
$\eta=\xi^{q_0+1}$. Then $\eta\in\GF(q)^*$, so it is sufficient to show that
$\xi\in\langle\eta\rangle\Gphi$. Since $\eta=\xi^{q_0+1}$ and
$\phi=\xi^{e/e_0}$, this subgroup contains all powers $\xi^k$ of
$\xi$ where $k$ is of the form
\[k=i(q_0+1)+je/e_0.\]
So we are done if we can show that $(q_0+1,e/e_0)=1$; since $e|q-1$ and
$q=q_0^t$ with $t$ odd, we have $(q_0+1,q-1)|2$ and we have to show that
$e/e_0$ is odd. This is evident in the case where $q$ is even, so we also
assume that $q$ is odd.
Write
\[q-1=2^r s, \qquad q_0-1=2^{r_0}s_0.\]
Now
\[q-1=q_0^t-1=(q_0-1)(1+q_0+\cdots + q_0^{t-1}),\]
and $q_0\equiv 1\bmod2$, hence
\[1+q_0+\cdots + q_0^{t-1}\equiv t\equiv 1\bmod 2;\]
we conclude that $r_0=r$. Moreover, $(q-1)/e$ is odd, so $e=2^rf$ with $f$ odd;
therefore $e_0=(e,q-1)$ is also divisible by $2^r$ and hence $e/e_0$ is indeed
odd.
\epf

Now in Theorem~2.4 of \cite{bn3}, it is shown that a nonstandard finite field
element of degree two over $\GF(q)$ with $q$-order $q+1$ is necessarily
primitive, that is, has order $q^2-1$.

\begin{cor}\label{Cmain}
A nonstandard element of degree two over a field $\GF(q)$ either is of type I,
with $q$-order 2 as in Example~\ref{E2}, or is of type II, with $q$-order of
the form $q_0+1\geq4$ with $q=q_0^t$ for an odd integer $t$, as 
in~Example~\ref{E3}.
\end{cor}

\section*{Acknowledgement}
We gratefully acknowledge email conversations with Brison and Nogueira in which
they pointed out that the nonexistance of nonstandard elements of degree two of
other types follows from our results in combination with some of their 
unpublished material.


\begin{thebibliography}{99}
%
\bibitem{bc} Thierry P.~Berger and Pascale Charpin, 
{\em The permutation group of
affine-invariant extended cyclic codes\/}, IEEE Trans.\ on Inform.\ Theory,
vol.\ 42, no.\ 6, November 1996, pp.\ 2194--2209.
%
%
\bibitem{ber} Thierry P.~Berger, {\em The automorphism group of
double-error-correcting BCH codes\/}, IEEE Trans.\ on Inform.\ Theory,
vol.\ 40, no.\ 2, March 1994, pp.\ 538--542.
%
\bibitem{bn1} Owen J.~Brison and J.\ Eurico Nogueira, {\em Linear recurring
sequence subgroups in finite fields}, Finite Fields Appl.\ 9
(2003), 413--422.
%
\bibitem{bn-comp} Owen J.~Brison and J.\ Eurico Nogueira, {\em Linear recurring
sequence subgroups in the complex field\/}, The Fibonacci Quarterly, vol.\ 41,
no.\ 5, Nov.\ 2003, pp.\ 397--404.
%
\bibitem{bn-mat} Owen J.~Brison and J.\ Eurico Nogueira,
{\em Matrices and Linear Recurrences in Finite Fields\/}, 
The Fibonacci Quarterly, vol.\ 44, no.\ 2, (2006), pp.\ 103--108.
%
\bibitem{bn2} Owen J.~Brison and J.\ Eurico Nogueira, {\em Second order
linear sequence subgroups in finite fields}, 
Finite Fields Appl., vol.\ 14, 2008, pp.\ 277--290.
%
\bibitem{bn3} Owen J.~Brison and J.\ Eurico Nogueira, {\em Second order
linear sequence subgroups in finite fields - II}, 
submitted to Finite Fields Appl.
%
\bibitem{cot}
P.J.~Cameron, G.R.~Omidi, B.~Tayfeh-Rezaie, {\em 3-Designs from $\PGL(2,q)$\/}, 
Electronic J.\ Combinatorics 13 (2006), \#R50 (11 p.).
%
\bibitem{cmot}
P.J.~Cameron, H.R.~Maimani,G.R.~Omidi, and B.~Tayfeh-Rezaie, 
{\em 3-Designs from $\PSL(2,q)$\/}, 
Discrete Math.\ 306 (2006), 3063--3073.
%
\bibitem{mcws}
F.J.~McWilliams, N.J.A.~Sloane, The Theory of Error-Correcting Codes, North
Holland, 1983.
%
\bibitem{dick} L.~Dickson, {Linear groups}, Dover.
%
\bibitem{hq}
Henk D.L.~Hollmann, Qing Xiang, 
{\em Association schemes from the action of
$\PGL(2,q)$ fixing a nonsingular conic in $\PG(2,q)$\/},
Journal of Algebraic Combinatorics, vol.\ 24, issue 2, Sept.\ 2006, 157--193.
%
\bibitem{hup} B.~Huppert, {Endliche gruppen I}, Springer.
%
\bibitem{ln} R.~Lidl, H.~Niederreiter, Finite fields, Addison-Wesley, 1983.
%
\bibitem{lint} J.H.~van~Lint, {Introduction to coding theory}, Graduate texts
in Mathematics 86, Springer-Verlag, 1992.
%
\bibitem{suz} M.~Suzuzki, {Group Theory I}, Springer.
%
\end{thebibliography}
\end{document}